\documentclass[journal]{IEEEtran}

\usepackage{graphicx}
\usepackage{epstopdf}
\usepackage{float}
\usepackage{subfig}

\usepackage{cite}
\usepackage{amsmath,amssymb,mathbbol}
 
\usepackage{mathtools}

\usepackage[active]{srcltx} 
\usepackage{color}
\usepackage{rotating}
\usepackage[colorlinks=true,linkcolor=blue,citecolor=blue]{hyperref}
\renewcommand{\QED}{\QEDopen}

\newtheorem{definition}{\bf Definition}
\newtheorem{theorem}{\bf Theorem}
\newtheorem{proposition}{\bf Proposition}
\newtheorem{lemma}{\bf Lemma}
\newtheorem{corollary}{\bf Corollary}
\newtheorem{remark}{\bf Remark}

\begin{document}
\title{Universal Weak Variable-Length Source Coding\\ on  Countably Infinite Alphabets}
\author{Jorge F. Silva and Pablo Piantanida
\thanks{This article has been accepted for publication by IEEE. Digital Object Identifier 10.1109/TIT.2019.2941895. Link: https://ieeexplore.ieee.org/document/8840879.}
\thanks{(c) 2019 IEEE. Personal use of this material is permitted.  Permission from IEEE must be obtained for all other uses, in any current or future media, including reprinting/republishing this material for advertising or promotional purposes, creating new collective works, for resale or redistribution to servers or lists, or reuse of any copyrighted component of this work in other works.}
\thanks{J. F. Silva was supported in part by CONICYTChile, Fondecyt, under Grant 1170854 and in part by the Advanced Center for Electrical and Electronic Engineering, Basal Project, under Grant FB0008. 
This project has received funding from the European Union’s Horizon 2020 research and innovation programme under the Marie Skłodowska-Curie grant agreement No 792464.
This work was presented in part at the 2016 and 2017 IEEE International Symposium on Information Theory (ISIT) [1], [2].
}
\thanks{J. F. Silva is with the Information and Decision Systems (IDS) Group, University of Chile, Santiago 412-3, Chile (e-mail: josilva@ing.uchile.cl).
P. Piantanida is with the Laboratoire des Signaux et Systèmes (L2S).}
\thanks{P. Piantanida is with the Laboratoire des Signaux et Systèmes (L2S), CentraleSupélec-CNRS-Université Paris-Sud, 91190 Gif-sur-Yvette, France, and also with the Montreal Institute for Learning Algorithms (Mila), Université de Montréal, Montréal, QC H3T 1J4, Canada (e-mail: pablo.piantanida@centralesupelec.fr).}
}

\maketitle

\begin{abstract}
Motivated from the fact that universal source coding on countably infinite alphabets ($\infty$-alphabets) is not feasible, this work introduces the  notion of ``almost lossless source coding''. Analog to the {\em weak variable-length source coding} problem studied by Han~(IEEE TIT, 2000, 46, 1217-1226), 
almost lossless source coding aims at relaxing the lossless block-wise assumption to allow an average per-letter distortion that vanishes asymptotically as the block-length tends to infinity. In this setup,  we show on one hand  that Shannon entropy characterizes the minimum achievable rate (similarly to the case of finite alphabet sources) while on the other that almost lossless universal source coding becomes feasible for the family of finite-entropy stationary memoryless sources with $\infty$-alphabets.  
Furthermore, we study a stronger notion of almost lossless universality that demands uniform convergence of the  average per-letter distortion to zero, where we establish a necessary and sufficient condition for the so-called family of ``envelope distributions'' to achieve it. Remarkably, this condition is the same necessary and sufficient condition needed for the existence of a strongly minimax (lossless) universal source code for the family of envelope distributions. Finally,  we show that an almost lossless coding scheme offers faster rate of convergence for the (minimax) redundancy compared to the well-known information radius developed for the lossless case at the expense of tolerating a non-zero distortion that vanishes to zero as the block-length grows. This shows that even when lossless universality is feasible, an almost lossless scheme can offer different regimes on the rates of convergence of the (worst case) redundancy versus the (worst case) distortion.
\end{abstract}

\begin{keywords}
Universal source coding,  countably infinite alphabets ($\infty$-alphabets), weak source coding, envelope distributions, information radius (i-radius), metric entropy analysis. 
\end{keywords}

\section{Introduction}
\label{sec_intro}
The problem of Universal Source Coding (USC) has a long history on information theory \cite{csiszar_2004,cover_2006,gyorfi_1994,davisson_1973,kieffer_1978}. This topic started with the seminal work of Davisson~\cite{davisson_1973} that formalizes the variable-length lossless coding and introduces relevant information quantities (mutual information and channel capacity~\cite{cover_2006}).  In lossless variable-length source coding,  it is well-known  that if we know the statistics of a source (memoryless or stationary and ergodic) the Shannon entropy (or Shannon entropy rate) provides the minimum achievable rate~\cite{cover_2006}. However, when the statistics of the source is not known  but it belongs to family of distributions $\Lambda$, then the problem reduces to characterize the worst-case expected overhead  (or \emph{worse-case redundancy}) that a pair of encoder and decoder experiences due to the lack of knowledge about true distribution governing the source samples to be encoded~\cite{csiszar_2004,gassiat_2018}.

A seminal information-theoretic result states that the least worst-case overhead (or minimax redundancy of $\Lambda$) is fully characterized by the \emph{information radius} of $\Lambda$~\cite{csiszar_2004}.  The information radius (i-radius) has been richly studied by the community and there are 
numerous contributions~\cite{boucheron_2009,bontemps_2014,haussler_1997,54897, 382017}, including applications to universal prediction of individual sequences~\cite{144706}. In particular,  it is well-known  that the  i-radius growths sub-linearly for the family of finite alphabet stationary and memoryless sources~\cite{csiszar_2004}, which implies the existence of an {universal source code} that achieves {Shannon entropy} for every distribution in  this family provided that the block length tends to infinity.
What is intriguing in this positive result obtained for finite alphabet memoryless sources is that it does not longer extend to the case of stationary and memoryless sources on {\em countably infinite alphabets} ($\infty$-alphabets), as was clearly shown in~\cite{kieffer_1978,gyorfi_1994,boucheron_2009}. From an information complexity perspective,  this infeasibility result implies that the i-radius of this family is unbounded for any finite block-length and, consequently,   lossless universal source coding for $\infty$-alphabet stationary and memoryless sources is an intractable problem.
%
In this regard, the proof presented by Gy\"{o}rfi {\em et al.} \cite[Theorem 1]{gyorfi_1994} is constructed over a connection between variable-length prefix-free codes and  distribution estimators, and the fact that  the redundancy of a given code upper bounds the expected divergence  between the true distribution and the induced  (through the code) estimate of the distribution. Then, the existence of an universal source code implies the existence of an universal estimator in the sense  of expected information divergence~\cite{barron_1992}.\footnote{Distribution estimator  consistent  in information divergence and reversed order information divergence were introduced by Barron {\em et al.} in \cite{barron_1992}. The connection between distribution estimation consistent in expected information divergence and universal source coding can be found in \cite{gyorfi_1994} and \cite[Sec. III.A]{barron_1992}.} The impossibility of achieving this learning objective for the family of finite entropy memoryless sources~\cite[Theorem 2]{gyorfi_1994} 
motives the main question addressed in this work that is, the study of a ``weak notion'' of universal variable-length source coding. 

In this framework, we propose to address the problem of universal source coding for $\infty$-alphabet stationary and memoryless sources by studying a weaker (lossy) notion of coding instead of the classical lossless definition \cite{csiszar_2004,cover_2006}. This  
notion borrows ideas from the seminal work by Han~\cite{han_2000} that allows reconstruction errors but assuming known statistic. In this paper, we investigate the idea of relaxing the lossless block-wise assumption with the goal that the corresponding weak universal source coding formulation will be reduced to a learning criterion that becomes feasible for the  whole family finite entropy stationary and memoryless sources on countably infinite alphabets. 
In particular,  we move from lossless coding to an  asymptotic vanishing distortion fidelity criterion based on the {\em Hamming distance}  as a fidelity metric. 

\subsection{Contributions}

Assuming that the distribution of the source is known, we first introduce the problem of ``almost lossless source coding'' for memoryless sources defined on countably infinite alphabets. Theorem~\ref{theorem_entropy} shows that {\em Shannon entropy}  characterizes the minimum achievable rate for this problem. The proof of this theorem adopts a result from Ho {\em et al.}~\cite{ho_2010} that provides a closed-form expression for the \emph{rate-distortion function} $R_\mu(d)$ on $\infty$-alphabets. From this characterization, we show that $\lim_{d \rightarrow 0} R_\mu(d)=\mathcal{H}(\mu)$ which is essential to prove this result\footnote{This result is well-known for finite alphabets, however the extension on countably infinite alphabets is not straightforward due to the  discontinuity of the entropy~\cite{silva_2012_isit,ho_2010b}.}.

Then, we address the problem of almost lossless universal source coding. The main difficulty arises in finding a lossy coding scheme that achieves asymptotically zero distortion, i.e.,  point-wise over the family, while guaranteeing that the worst-case average redundancy --w.r.t. the minimum achievable rate-- vanishes with the block-length~\cite{csiszar_2004}. The proof of existence of an universal code with the desired property relies on a two-stage coding scheme that first quantizes (symbol-by-symbol) the $\infty$-alphabet and then applies a lossless variable-length code over the resulting quantized symbols. Our main result, stated in Theorem~\ref{th_achie_almost_lossless_universality}, shows that almost lossless universal source coding is feasible for the family of finite entropy stationary and memoryless sources.

We further  study the possibility of obtaining rates of convergence for the worst-case distortion and the worst-case redundancy. To this end, we restrict our analysis to the family of stationary and memoryless sources with 1D-densities dominated by an envelope function $f$, which was previously studied in~\cite{boucheron_2009,bontemps_2014,bontemps_2011,acharya_2014}. Theorem~\ref{th_dichotomy_stronger_universality} presents a necessary and sufficient condition on $f$ to achieve an uniform convergence (over the family) of the distortion to zero and, simultaneously, a vanishing worst-case average redundancy. Remarkably, this condition ($f$ being a summable function) is the same necessary and sufficient condition needed for the existence of a strongly minimax (lossless) universal source code~\cite[Theorems 3 and 4]{boucheron_2009}. 

Finally, we provide an analysis of the potential benefit of an almost lossless two-stage coding scheme by exploring the family of  envelope distributions that admits strong minimax universality in lossless source coding \cite{boucheron_2009,gassiat_2018}.  In this  context,  Theorem~\ref{th_information_radius_gain_tail_schemme} shows that we can have an almost lossless approach that offers a non-trivial reduction to the rate of convergence of the worst-case redundancy, with respect to the well-known i-radius developed for the lossless case,  at the expense of tolerating a non-zero distortion that vanishes with the blocklength. This result provides evidence that even in the case where lossless universality is feasible, an almost lossless scheme can reduce the rate of the worst-case redundancy and consequently, it offers ways of achieving different regimes for the rate of convergence of the redundancy versus the distortion. The proof of this result uses advanced tools by Haussler and Opper~ \cite{haussler_1997} to relate the minimax redundancy of a family of distributions with its metric entropy with respect to the {\em Hellinger distance}. Indeed, this metric entropy approach has shown to be instrumental to derive tight bounds on the i-radius for summable envelope distributions in~\cite{bontemps_2014}. We extended this metric entropy approach to our almost lossless coding setting with a two-stage coding scheme to characterize the precise regime in which we can achieve gains in the rate of convergence of the redundancy.

\subsection{Organization of the Paper}
The rest of the paper is organized as follows. Section~\ref{sec_background} introduces  some definitions and preliminary results. 
Section~\ref{subsec:alsc} introduces our main weak source coding problem and shows that Shannon entropy is the minimum achievable rate  
provided that the statistics of the source is known. Section~\ref{sec_universal} presents the problem of almost lossless universal source coding and proves its feasibility for the family of finite entropy memoryless distributions on $\infty$-alphabets. Section~\ref{sub_sec:uniform_convergence_distortion} elaborates a result for a stronger notion of almost lossless universality, and Section~\ref{sec_redundancy_gains_envelop_families} studies the gains in the rate of convergence of the minimax-redundancy that can be obtained with an almost lossless scheme  for families of distributions that admit lossless USC.  Finally, Section~\ref{sec_discusion}  concludes with a summary of the work. 
The proofs of the main results of this paper are presented in Section~\ref{section_proof}, while some supporting results are relegated to the Appendix section.

\subsection{Basic Notation}
The following notations and conventions are used throughout this article: 
 $(x_n)_n$ will denote an infinite dimensional sequence in $\mathbb{R}^{\mathbb{N}}$; 
 $(x_n)_n \ll (y_n)_n$ or,  alternatively,  $(x_n)_n$ being $o(y_n)$,  means that  $\lim_{n \rightarrow \infty} \frac{x_n}{y_n}=0$;
 $(x_n)_n$ being $o(1)$ means that  $\lim_{n \rightarrow \infty}  {x_n}=0$;
  $(x_n)_n \leq (y_n)_n$  means that  $x_n\leq y_n$ for all $n\geq 1$; 
  $(x_n)_n \leq (y_n)_n$ eventually in $n$ means that there exists $N>1$ such that  $x_n\leq y_n$ for all $n\geq N$; 
  $(x_n)_n \sim (y_n)_n $ means that $\lim_{n \rightarrow \infty} \frac{x_n}{y_n}=1$; 
  $\mathcal{X}$ will denote the countably infinite alphabet and $\mathcal{P}(\mathcal{X})$ the collection of
  probability measures in   $\mathcal{X}$;
  for any function $f:\mathcal{X} \longrightarrow \mathbb{R}$, $f$ is said to be summable (denoted by $f\in \ell_1(\mathcal{X})$) if $\sum_{n\in \mathcal{X}} \left| f(x)\right| < \infty$; 
 and for $\mu, v\in \mathcal{P}(\mathcal{X})$,  $ \mu \ll v$ means that 
 if $v(B)=0$ then $\mu(B)=0$ for any $B\subset \mathcal{X}$.

\section{Preliminaries} 
\label{sec_background}
This section introduces some useful concepts, definitions and results that will be needed across the paper. Let ${\bf X}= \left\{X_i \right\}_{i=1}^\infty$ be a stationary and memoryless process (or i.i.d. source) with values in a countably infinite alphabet ($\infty$-alphabets) $\mathcal{X}$ equipped with a probability measure $\mu$ defined on the measurable space $(\mathcal{X}, \mathcal{B}(X))$\footnote{$\mathcal{B}(\mathcal{X})$ denotes the power set of $\mathcal{X}$.}.  Let $X^n=(X_1,\dots,X_n)$ denote a finite block of length $n$ of the process following the product  measure $\mu^n$ on $(\mathcal{X}^n, \mathcal{B}(X^n))$\footnote{The product measure satisfies the memoryless condition for all $\mathcal{B}_1\times \cdots \times \mathcal{B}_n \in \mathcal{X}^n$ then $\mu^n(\mathcal{B}_1\times \cdots \times \mathcal{B}_n)=\mu( \mathcal{B}_1)\cdots \mu( \mathcal{B}_n)$.}.  
Let us denote by $\mathcal{P}(\mathcal{X})$ the family of probability measures in $(\mathcal{X}, \mathcal{B}(X))$, where for every $\mu \in \mathcal{P}(\mathcal{X})$, we understand $f_\mu(x) \coloneqq    \frac{d \mu}{d\lambda}(x) = \mu(  \left\{ x\right\})$ to be a short-hand for its probability mass function (pmf). Let $\textrm{supp}(f)= \left\{x \in \mathcal{X}: \left| f(x)\right|>0\right\}$ and let $\mathcal{P}_\mathcal{H}(\mathcal{X})\coloneqq     \left\{ \mu: H(\mu)<\infty \right\} \subset \mathcal{P}(\mathcal{X})$ denote the collection of finite Shannon entropy probabilities~\cite{shannon_1948} where
\begin{equation}\label{eq_sec_pre_1}
	H(\mu) = -\sum_{x\in \mathcal{X}} f_\mu(x) \log f_\mu(x),  
\end{equation}
with $\log$ function on base 2.  

Given an i.i.d. source ${\bf X}= \left\{X_i \right\}_{i=1}^\infty$  with distribution $\mu\in \mathcal{P}(\mathcal{X})$,  let us consider a (variable length) lossless code  $f_n$  of length $n$ as a prefix-free mapping from $\mathcal{X}^n$ to finite sequences of symbols in $\left\{0, 1 \right\}$ \cite{cover_2006}.
It is well-known that $\mathbb{E}_{X^n}\left\{ \mathcal{L} (f_n(X^n)) \right\} \geq H(\mu^n)$ \cite{cover_2006}, where $\mathcal{L}(\cdot)$ indicates the functional that returns the length of binary sequences in $\left\{0,1\right\}^\star \coloneqq \cup_{k \geq 1} \left\{0, 1 \right\}^k$.
Then, the average length (in bits) used to encode $X^n$ with $f_n$  can be measured relative to $H(\mu^n)$,  which motivates the introduction of the 
average redundancy (or redundancy) of $f_n$ by the expression: $\mathbb{E}_{X^n}\left\{ \mathcal{L} (f_n(X^n)) \right\} - H(\mu^n)$. When $\mu$ is known,  
the {\em Huffman code} uses that information to offer an optimal prefix-free mapping (minimizing the average code-length)  
whose redundancy is upper bound by 1 \cite{cover_2006}.
%

\subsection{Strong Minimax Universality, Information Radius and Envelope Families}
\label{sub_sec_envelop_results_usc}
In universal source coding (USC), we need to encode a stationary memoryless source {\bf X} with an unknown probability $\mu$ that belongs to 
a class of models $\Lambda \subset \mathcal{P}(\mathcal{X})$.  In this context,  a natural performance for a prefix-free encoder $f_n:\mathcal{X}^n \rightarrow \left\{0,1\right\}^\star$ is the worse-case  (over the family $\Lambda$) redundancy expressed by:  
$$R(f_n, \mu^n) \coloneqq \sup_{\mu^n \in \Lambda^n} \left(  \mathbb{E}_{X^n\sim \mu^n }\left\{ \mathcal{L} (f_n(X^n)) \right\} - H(\mu^n) \right),$$  
where $\Lambda^n \coloneqq \left\{\mu^n, \mu\in \Lambda \right\}\subset \mathcal{P}(\mathcal{X}^n)$ is a short-hand for the family of $n$-fold (product) distributions
induced by $\Lambda$. 
This worse-case performance indicator motivates the adoption of the minimax design principle:  
 $\min_{f_n} R(f_n, \mu^n)$ frequently used in USC \cite{csiszar_2004}, where the optimization is carried over the family of prefix-free codes. 
Importantly, there is a 
well-documented correspondence between prefix-free codes for $\mathcal{X}^n$ and probabilities in $\mathcal{P}(\mathcal{X}^n)$ \cite{cover_2006}.
In fact, the {\em Kraft-MacMillan inequality} defines a probability in  $\mathcal{P}(\mathcal{X}^n)$ from a prefix-free code of length $n$ \cite{cover_2006},
and conversely arithmetic coding provides a prefix-free code for  $\mathcal{X}^n$ from a probability $v\in \mathcal{P}(\mathcal{X}^n)$,
where the length of this code  (in bits) is $\lceil -\log v(x^n)  \rceil  + 1$ for any $x^n\in \mathcal{X}^n$ \cite{rissanen_1979,cover_2006}. Then  
the minimax redundancy problem for USC reduces to the solution of the i-radius problem \cite{csiszar_2004}\footnote{In fact, it follows that $R^+(\Lambda^n) +2 \geq \min_{f_n} R(f_n, \mu^n) \geq R^+(\Lambda^n)$ \cite{cover_2006}.}:
\begin{equation}\label{eq_sec_pre_2}
	R^+(\Lambda^n) \coloneqq   \inf_{v \in  \mathcal{P}(\mathcal{X}^n)} \sup_{\mu^n \in \Lambda^n} \mathcal{D}(\mu^n | v)
\end{equation}
and 
\begin{equation}\label{eq_sec_pre_3}
	\mathcal{D}(\mu^n | v) = \sum_{x^n\in \mathcal{X}^n} f_{\mu^n}(x^n) \log \frac{f_{\mu^n}(x^n) }{f_{v}(x^n)}
\end{equation}
is the divergence of $\mu^n$ with respect to $v$ \cite{kullback1958,csiszar_2004,cover_2006}. 
Again using this connection between codes and distributions, a class $\Lambda \subset \mathcal{P}(\mathcal{X})$
of i.i.d. sources will be said to be {\em weakly universal} if there is a sequence of probabilities $(v_n)_n$ (where $v_n\in \mathcal{P}(X^n)$ for all $n$)
such that $\sup_{\mu \in \Lambda}  \lim_{n \rightarrow \infty} \frac{1}{n} \mathcal{D}(\mu^n | v_n)=0$, and  it will be {\em strongly
universal} (or strongly minimax universal) if  $\lim_{n \rightarrow \infty} \sup_{\mu \in \Lambda} \frac{1}{n} \mathcal{D}(\mu^n | v_n)=0$.
For the last stringent USC objective,  the minimax redundancy sequence $(R^+(\Lambda^n))_n$ of $\Lambda$ in (\ref{eq_sec_pre_2}) 
determines if the family is strongly minimax universal \cite{csiszar_2004,gassiat_2018}. For $\infty$-alphabets i.i.d. sources, it is known that $R^+(\mathcal{P}(X)^n)=\infty$ and, furthermore,  weak minimax universality is not feasible \cite{kieffer_1978,gyorfi_1994,boucheron_2009}. This motivates the study of strong minimax universality over  sub-collections  of distributions whose 1D densities are dominated by an envelope function \cite{boucheron_2009,bontemps_2014,bontemps_2011}: 
\begin{definition}\label{def_envelop}
	Given a non-negative function $f:\mathcal{X} \longrightarrow \mathbb{R}^+$, the \emph{envelope family} indexed by $f$ is given by:
	\begin{equation}\label{eq_ucd_1}
		\Lambda_{f} \coloneqq    \big\{\mu\in \mathcal{P}(\mathcal{X}): f_\mu(x)\leq f(x),\  \textrm{ for }\ x\in\mathcal{X} \big\}.
	\end{equation}
\end{definition}
The next result by Boucheron {\em et al.}~\cite{boucheron_2009} establishes a necessary and sufficient condition to 
make $\Lambda_{f}$ strongly minimax universal:
\begin{theorem}\label{dichotomy_envelops} 
	\cite[Ths. 3 and 4]{boucheron_2009}
	Let us consider $f: \mathcal{X} \rightarrow \mathbb{R}^+$, with $\mathcal{X}$ an $\infty$-alphabet, and the family of i.i.d. envelope distributions $\left\{ \Lambda^n_f, n\geq 1 \right\}$. It follows that: 
	\begin{itemize}
	\item[i)] If $f$ is summable, i.e., $f\in \ell_1(\mathcal{X})$,  then $R^{+}(\Lambda_f^n)<\infty$ for all $n \geq 1$, and furthermore $(R^{+}(\Lambda_f^n))_n$ is $o(n)$. 
	\item[ii)] Otherwise, $R^{+}(\Lambda_f^n)=\infty$ for all $n\geq 1$.
	\end{itemize}
\end{theorem} 
Therefore for any envelope family in an $\infty$-alphabet,  either it is strongly minimax universal (i.e.,$(R^{+}(\Lambda_f^n))_n$ is $o(n)$) 
or its  i-radius in (\ref{eq_sec_pre_2}) is unbounded for any finite block-length. The last unbounded scenario means that for any $n$  and any prefix-free code $f_n$ there is a distribution $\mu$ in the family 
where the average redundancy of $f_n$ (with respect to the entropy $H(\mu^n)$) is unbounded. Furthermore for a summable envelope function $f$, 
a series of relevant results stipulate the way $(R^{+}(\Lambda_f^n)/n)_n$ tends to zero function of specific tail attributes of $f$ \cite{boucheron_2009,bontemps_2014,bontemps_2011}. We select a result here that will be important for our exposition, for  which some definitions are needed:
\begin{definition}  (Bontemps {\em et al.}\cite{bontemps_2014}) 
\label{def_tail_function_critical_dimension}
For a non-negative envelope function $f\in \ell_1(\mathcal{X})$ with  $\left|\textrm{supp}(f)\right|=\infty$,  we can determine 
$$l_{f} \coloneqq     \max \left\{k: \sum_{j\geq k} f(j) \geq 1\right\}$$ 
and the associated envelope probability $\mu_f \in \Lambda_f$ by: 
\begin{equation}
\mu_f( \left\{ k\right\}) \coloneqq \left\{
\begin{array}{lll}
0, & \textrm{ for } & k<l_{f} \\
f(k), & \textrm{ for } &  k>l_{f}\\
1- \sum\limits_{j > l_f} f(j), & \textrm{ for } & k=l_f. 
\end{array}\right.
\end{equation} 
\end{definition}

\begin{definition}(Bontemps {\em et al.}\cite{bontemps_2014}) 
\label{def_tail_function_critical_dimension_b}
If we denote by ${F}_f(u) \coloneqq \mu_f( \left\{1,\dots, u\right\})$ the envelope distribution and by  $\bar{F}_f(u) \coloneqq 1- F_f(u)$  the tail function of $f$ (for all $u\geq 1$), we can define the quantile of order $\frac{1}{n}$ of $\mu_f$ as the solutions of \cite{bontemps_2014}:
\begin{equation}\label{eq_sec_rgef_4}
	u^*_f(n) \coloneqq    \min  \left\{ u\geq 1:  \bar{F}_f(u) <\frac{1}{n} \right\} \text{ for all } n \geq 1.
\end{equation}
\end{definition}
\begin{theorem}\label{critimal_dimension_envelop}  \cite[Th. 4]{boucheron_2009}\&\cite[Th. 2, Prop. 5]{bontemps_2014}
	Let us consider the envelope family  $\left\{\Lambda^n_f, n \geq 1 \right\}$ in Def. \ref{def_envelop}  with $f\in \ell_{1}(\mathcal{X})$. 
	Then there is a sequence $(\xi_n)_n$ being $o(1)$ such that
	\begin{align*}
	(1+ \xi_n) \frac{ (u^*_f(n)-1) }{4} \log n &\leq R^+(\Lambda_f^n) \\
								 &\leq 2+ \log e + \frac{(u^*_f(n)-1)}{2} \log n
	\end{align*}
	holds eventually with $n$. 
\end{theorem} 
It has been shown that when $f\in \ell_{1}(\mathcal{X})$ then $(u^*_f(n) \log n)_n$ is $o(n)$ \cite{bontemps_2014}, therefore Theorem \ref{critimal_dimension_envelop}  
is consistent with Theorem \ref{dichotomy_envelops}. Importantly, $(u^*_f(n))_n$ captures the complexity of the envelope family by 
determining the worse-case redundancy (and its velocity of convergence to zero with $n$) that an optimal universal code can achieve 
in compressing (losslessly) a stationary and memoryless source in this family.
 
\section{Almost Lossless Source Coding}
\label{subsec:alsc}
We now introduce the notion of a lossy variable-length coding of $n$ source symbols, which consists of a pair  $(f_n, g_n)$ where $f_n: \mathcal{X}^n \longrightarrow \left\{0,1\right\}^\star$ is a prefix free variable-length code (encoder) \cite{cover_2006} and $g_n: \left\{0,1\right\}^\star \longrightarrow \mathcal{X}^n$ is the inverse mapping from bits to source symbols (decoder). Inspired by the weak coding setting introduced by Han \cite{han_2000}, the possibility that $\left\{x^n: g_n(f_n(x^n))\neq x^n \right\} \neq \emptyset$ is allowed. In order to quantify the loss induced by this encoding process, a per letter distortion measure characterization $\rho:\mathcal{X} \times \mathcal{X}: \longrightarrow \mathbb{R}^+$ is considered \cite{berger_1998,gray_1990}, where  for $x^n,y^n\in \mathcal{X}^n$ the distortion is given by
\begin{equation}\label{eq_sec_background_3}
	\rho_n(x^n, y^n)\coloneqq \frac{1}{n}\sum_{i=1}^n \rho(x_i,y_i). 
\end{equation}
Given an information source $\mathbf{X}=\left\{X_i \right\}_{i=1}^n$, the average 
distortion induced  by the pair $(f_n, g_n)$ is
\begin{equation}\label{eq_sec_background_4a}
	d(f_n, g_n,\mu^n) \coloneqq    \mathbb{E}_{X^n\sim \mu^n} \big\{  \rho_n\big(X^n, g_n(f_n(X^n))\big) \big\}.
\end{equation}
For the rest of the paper, we will focus on the special case where $\rho(x,y)=\mathbf{1}_{\left\{ x\neq y \right\} }$. 
Then, $\rho_n(x^n, y^n)$ is the normalized {\em Hamming distance} between the sequences $(x^n, y^n)$. 
On the other hand, the rate of the pair $(f_n,g_n)$ (in bits per sample) is
\begin{equation}\label{eq_sec_background_5a}
	r(f_n, \mu^n) \coloneqq    \frac{1}{n} \mathbb{E}_{X^n\sim \mu^n} \left\{ \mathcal{L}(f_n(X^n)) \right\}.
\end{equation}
At this stage,  we can 
introduce the almost-lossless source coding problem and with this,  the standard notion of minimum achievable rate. 

\begin{definition}[Achievability]\label{def_achievable}
	Given an information source  $\mathbf{X}= \left\{X_i \right\}_{i=1}^\infty$, we say that a rate $R>0$ is achievable
	for almost-losslessly encoding $\mathbf{X}$, i.e., with zero asymptotic distortion, if there exists a sequence of encoder and decoder mappings $\left\{ (f_n, g_n)\right\}_{n\geq 1}$ satisfying: 
	\begin{IEEEeqnarray}{rCl}
		\limsup_{n \longrightarrow \infty} r(f_n, \mu^n) &\leq& R, \\ 
		\lim_{n \longrightarrow \infty} d(f_n, g_n,\mu^n) &=&0.
	\end{IEEEeqnarray}
	The minimum achievable rate is then defined as: 
	\begin{align}\label{eq_sec_background_7}
		R_{al}(\mathbf{X})\coloneqq   \min  \left\{ R: R  \text{ is achievable for } \mathbf{X}\right\}. 
	\end{align}
\end{definition}

Let $R_{al}(\mu)$  denotes the minimum achievable rate of a stationary and memoryless source driven by $\mu\in \mathcal{P}(\mathcal{X})$. The next theorem characterizes $R_{al}(\mu)$ provided that the source statistics is known.

\begin{theorem}[Known statistics] \label{theorem_entropy}
Given a stationary and memoryless source on a $\infty$-alphabet 
driven by the probability measure $\mu \in \mathcal{P}_{\mathcal{H}}(\mathcal{X})$, it follows that $R_{al}(\mu)=H(\mu)$.	
\end{theorem}

The proof is presented in Section \ref{proof_theorem_entropy}. 

As it is expected, Shannon entropy characterizes the minimum achievable rate for the almost lossless source coding problem formulated in Definition~\ref{def_achievable}. In the proof of  Theorem \ref{theorem_entropy}, we adopt a result from Ho {\em et al.} \cite{ho_2010} that provides a closed-form expression for the rate-distortion function $R_\mu(d)$ of $\mu$ on $\infty$-alphabet 
through a tight upper bound on the conditional entropy for a given minimal error probability \cite[Theorem 1]{ho_2010}. 
From this characterization, we show that $\lim_{d \rightarrow 0} R_\mu(d)=H(\mu)$, which is essential to show the  
result\footnote{Theorem \ref{theorem_entropy} is well-known for finite alphabet stationary memoryless sources, however its extension to countably infinite alphabets is not straightforward  due to the discontinuity of the entropy. The interested reader may be refer to~\cite{ho_2010b,silva_2012_isit,silva_2018} for further details.}.  

\subsection{A Two-Stage Source Coding Scheme} \label{sub_sec:quantization_alphabet}
In this section, we consider a two-stage source coding scheme that first applies a lossy  (symbol-wise)  reduction of the alphabet, and second a variable-length lossless source code over the restricted alphabet.
Let us define the finite set $\Gamma_k \coloneqq    \left\{1,\dots,k\right\}$.   We say that a {\em two-stage lossy code} of block-length $n$  and size $k$ is the composition of:  a lossy mapping of the alphabet, represented by a pair of functions  $(\phi_n, \psi_n)$, where $\phi_n:\mathcal{X} \longrightarrow \Gamma_k$ and $\psi_n: \Gamma_k \longrightarrow \mathcal{X}$, and a fixed to variable-length prefix-free pair of lossless encoder-decoder $(\mathcal{C}_n, \mathcal{D}_n)$, where $\mathcal{C}_n:\Gamma_k^n \longrightarrow \left\{0,1\right\}^\star$ and $\mathcal{D}_n: \left\{0,1\right\}^\star \longrightarrow \Gamma_k^n$.
 
Given a source $\mathbf{X}= \left\{X_i \right\}_{i=1}^\infty$ and an $(n,k_n)$-lossy source encode $(\phi_n, \psi_n,\mathcal{C}_n, \mathcal{D}_n)$\footnote{For brevity, the decoding function $\mathcal{D}_n: \left\{0,1\right\}^\star \longrightarrow  \Gamma_k^n$ will be omitted and considered implicit in the rest of the exposition.}, the lossy encoding of $\mathbf{X}$ induced by $(\phi_n, \psi_n,\mathcal{C}_n)$ is a two-stage process where first a quantization of size $k_n$ over $\mathcal{X}^n$ is made (letter-by-letter) to generate a finite alphabet random sequence $Y^n\coloneqq  (\phi_n(X_1),\dots, \phi_n(X_n))$ and then, a variable-length coding is applied to produce $\mathcal{C}_n(Y^n)$.  Associated to the pair $(\phi_n, \psi_n)$, there is an induced partition of $\mathcal{X}$ given by: 
\begin{equation}\label{eq_sec_background_2}
	\pi_n\coloneqq    \left\{ \mathcal{A}_{n,i} \coloneqq    \phi_n^{-1}(\left\{ i\right\})\,| \, i \in \Gamma_{k_n} \right\} \subset \mathcal{B}(\mathcal{X}), 
\end{equation}
and a collection of prototypes\footnote{Without loss of generality, we assume that $y_{n,i}\in \mathcal{A}_{n,i}$.} $\big\{y_{n,i}\coloneqq    \psi_n(i)\in \mathcal{A}_{n,i}\,| \, i \in \Gamma_{k_n} \big\} \subset \mathcal{X}$. 
The resulting distortion incurred by this code is given by
\begin{equation}\label{eq_sec_background_4}
	d(\phi_n, \psi_n,\mu^n) \coloneqq    \mathbb{E}_{X^n\sim \mu^n} \big\{  \rho_n\big(X^n, \Psi_n(\Phi_n(X^n))\big) \big\},
\end{equation}
where $\hat{X}^n = \Psi_n(\Phi_n(X^n))$ is a short-hand to denote $\big(\psi_n(\phi_n(X_1)),\ldots, \psi_n(\phi_n(X_n))\big)$.
%
On the other hand, the coding rate is:
\begin{equation}\label{eq_sec_background_5}
	r(\phi_n, \mathcal{C}_n, \mu^n) \coloneqq    \frac{1}{n} \mathbb{E}_{X^n\sim \mu^n} \big\{ \mathcal{L}\big(\mathcal{C}_n(\Phi_n(X^n))\big) \big\},
\end{equation}
with $\Phi_n(X^n)$ denoting $(\phi_n(X_1),\dots,\phi_n(X_n))$. An illustration of this two-stage process is presented 
in Figure \ref{fig1}.
\begin{figure*}[h]
\centering
\includegraphics[width=0.9\textwidth]{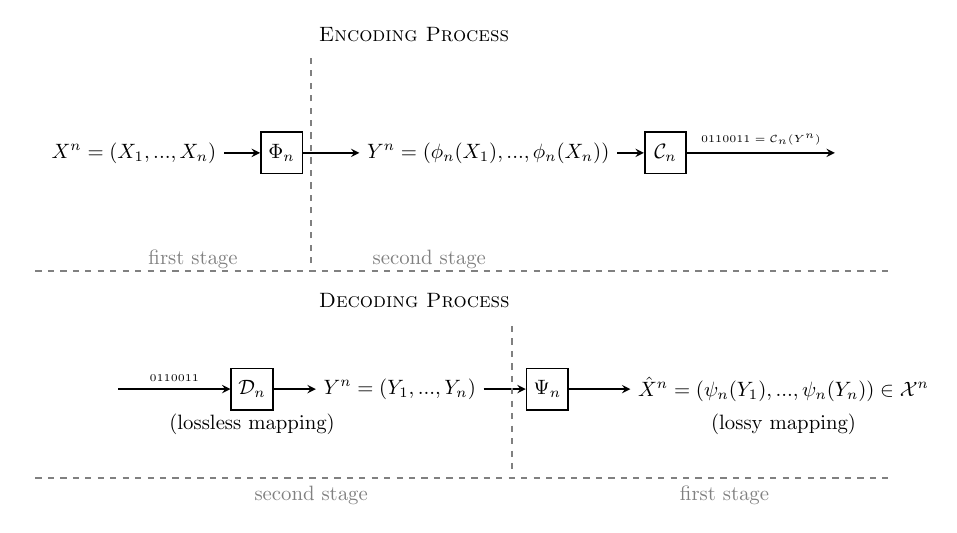}
\caption{Illustration of the two-stage lossy coding scheme $(\phi_n, \psi_n,\mathcal{C}_n, \mathcal{D}_n)$ presented in Section \ref{sub_sec:quantization_alphabet}.}
\label{fig1}
\end{figure*}

At this point, it is worth mentioning some basic properties on the partitions induced by $(\phi_n, \psi_n,\mathcal{C}_n)$ on $\mathcal{X}$.
\begin{definition}\label{def_asymtotically_suff_part}
	A sequence of partitions $\left\{ \pi_n\right\}_{ n \geq 1}$ of $\mathcal{X}$ is said to 
	be asymptotically sufficient with respect to $\mu\in \mathcal{P}(\mathcal{X})$, if for all $x\in \textrm{supp}(\mu)$
	\begin{equation}\label{eq_proof_lemma_entropy_1}
		\lim_{n \rightarrow \infty } \pi_n(x)=   \left\{ x\right\}, \ \mu\text{-almost everywhere}, 
	\end{equation}
	where $\pi_n(x)\subset \mathcal{X}$ denotes the cell in $\pi$ that contains $x$ and
	the almost-sure limit with respect to $\mu$ stated in (\ref{eq_proof_lemma_entropy_1}) 
	refers to the condition: 
	$$\mu\left(\limsup\limits_{n \rightarrow \infty} \pi_n(x) \setminus  \left\{ x\right\}\right)=0,$$ 
	which is equivalent to $\lim\limits_{n  \rightarrow \infty} \mu\left(\bigcup_{k\geq n} \pi_k(x) \setminus  \left\{ x\right\}\right)=0$.
\end{definition}

Consider now almost lossless coding for which we can state the following.
\begin{lemma}\label{lemma_zero_distortion}
	Let $\mathbf{X}$ be a stationary and memoryless source  driven by $\mu$. A necessary and sufficient condition 
	for $\left\{ (\phi_n, \psi_n,\mathcal{C}_n): n\geq 1\right\}$  to have that 
		$\lim\limits_{n \rightarrow \infty} d(\phi_n, \psi_n,\mu^n) =0$
	is that $\left\{\pi_n \right\}_{n\geq 1}$ in (\ref{eq_sec_background_2}) is asymptotically sufficient for $\mu$.
\end{lemma}

The proof of Lemma \ref{lemma_zero_distortion} is presented in Appendix~\ref{proof_lemma_zero_distortion}.

Studying the minimum achievable rate for zero-distortion coding requires the following definition.

\begin{definition}\label{def_entropy_sub_sigmafield}
	For $\mu\in \mathcal{P}_{\mathcal{H}}(\mathcal{X})$ and a partition $\pi$ of  
	$\mathcal{X}$, the entropy of $\mu$ restricted to the sigma-field induced by  $\pi$, 
	which is denoted by $\sigma(\pi)$, is given by  
	\begin{equation}\label{eq_proof_lemma_entropy_2}
		H_{\sigma(\pi)}(\mu) \coloneqq    - \sum_{\mathcal{A}\in \pi}  \mu(\mathcal{A}) \log \mu(\mathcal{A}).
	\end{equation}
\end{definition}
A basic  inequality~\cite{csiszar_2004,cover_2006} shows that if $\sigma(\pi) \subset \sigma(\bar{\pi})$, then $H_{\sigma(\pi)}(\mu) \leq H_{\sigma(\bar{\pi})}(\mu)$ for every $\mu$. In particular, $H_{\sigma(\pi)}(\mu) \leq H(\mu)$, where it is simple to show that $H(\mu)=\sup_{\pi\in \Pi(\mathcal{X})} H_{\sigma(\pi)}(\mu)$  with $\Pi(\mathcal{X})$ representing the collection of finite partitions of $\mathcal{X}$. 
Furthermore, it is possible to state the following result.
\begin{lemma}\label{lemma_asym_suff}
	If a sequence of partitions $\left\{ \pi_n\right\}_{n \geq 1}$ of $\mathcal{X}$ is asymptotically 
	sufficient with respect to $\mu$ (Def. \ref{def_asymtotically_suff_part}), then 
	\begin{equation}\label{eq_proof_lemma_entropy_3}
		\lim_{n \longrightarrow \infty} H_{\sigma(\pi_n)}(\mu)=H(\mu).
	\end{equation}
\end{lemma}

The proof of this result is presented in Appendix~\ref{proof_lemma_asym_suff}.

This implies that if a two-stage scheme $\left\{ (\phi_n, \psi_n,\mathcal{C}_n): n\geq 1\right\}$ achieves zero distortion, then   
\begin{equation}
\left\{ \pi_n=\left\{\phi_n^{-1}(\left\{ i\right\}): i \in \Gamma_k \right\}: n\geq 1 \right\}
	\end{equation}
is asymptotically sufficient for $\mu$ (cf. Lemma~\ref{lemma_zero_distortion}). From the well-known result in lossless variable-length source coding~\cite{cover_2006}, we have that:
	\begin{IEEEeqnarray}{rCl}
		r(\phi_n, \mathcal{C}_n, \mu^n) &\geq &\frac{1}{n} H\left(\Phi_n(X^n)\right) \nonumber\\
		&=& H_{\sigma(\pi_n)}(\mu)\label{eq_proof_lemma_entropy_4}
\end{IEEEeqnarray}
and consequently, Lemma~\ref{lemma_asym_suff} implies that
		$\liminf\limits_{n \rightarrow \infty} r(\phi_n, \mathcal{C}_n, \mu^n) \geq H(\mu)$.
Hence, letting $\bar{R}_{al}(\mu)$ to be the minimum achievable rate w.r.t. the family of two-stage lossy schemes in  Definition~\ref{def_achievable}, we obtain that  $\bar{R}_{al}(\mu) \geq  R_{al}(\mu) =H(\mu)$.

The next result shows that there is no additional  overhead  (in terms of bits per sample),  if we restrict the problem to the family of two-stage lossy schemes.
\begin{proposition}\label{proposition_entropy_two_stage} 
For a stationary and memoryless source $\mathbf{X}=\left\{X_i \right\}_{ i=1}^\infty$ driven by $\mu \in  \mathcal{P}_{\mathcal{H}}(\mathcal{X})$, 
	\begin{equation}
		\bar{R}_{al}(\mu)=R_{al}(\mu) =H(\mu).	\nonumber
	\end{equation}
\end{proposition}

The proof is presented in Appendix~\ref{proof_proposition_entropy_two_stage}.

\section{Universal  Almost Lossless Source Coding}
\label{sec_universal}
Consider a stationary and memoryless source $\left\{X_n\right\}_{n=1}^\infty$  on a $\infty$-alphabet  
with unknown distribution but belonging to a family $\Lambda \subset \mathcal{P}(\mathcal{X})$.
The main question to address here is if there exists a lossy coding scheme whose rate achieves the minimum feasible rate in~Theorem~\ref{theorem_entropy},  for every possible distribution in   
$\Lambda$, while the  distortion goes to zero as the block-length tends to infinity as defined below.
\begin{definition}
	\label{def_universality}
	A family of distribution $\Lambda$ is said to admit an almost lossless USC
	scheme, if there is a lossy source code $\left\{ (f_n, g_n)\right\}_{n\geq 1}$ simultaneously satisfying:
	\begin{equation} \label{eq_sec_background_8}
		\sup_{\mu\in \Lambda} \lim_{n \rightarrow \infty} d(f_n, g_n,\mu^n) =0,
\end{equation}
and
\begin{equation} 
		\label{eq_sec_background_8b}
		\lim_{n \rightarrow \infty} \sup_{\mu\in \Lambda} \Big(r(f_n, \mu^n) -H(\mu)\Big)=0.
\end{equation}
\end{definition}
An almost lossless universal code provides a point-wise convergence of the distortion to zero for every $\mu \in \Lambda$ while constraining the worst-case expected redundancy to vanish as the block length tends to infinity. It is obvious from Definition~\ref{def_universality} that if $\Lambda$  admits a classical lossless universal source code~\cite{kieffer_1978,davisson_1973}, i.e., the worst-case average redundancy vanishes with zero distortion for every finite $n$, then it admits an almost lossless USC. 
The next result shows that there is a richer family of distributions that admits an almost lossless USC scheme:

\begin{theorem}[Feasibility] \label{th_achie_almost_lossless_universality}
	The family $\mathcal{P}_{\mathcal{H}}(\mathcal{X})$ admits an almost lossless USC scheme.
\end{theorem}

The proof  is presented in Section \ref{proof_th_achie_almost_lossless_universality}.

Remarkably,  Theorem~\ref{th_achie_almost_lossless_universality} shows that a weak notion of universality allows to code the complete collection of finite entropy stationary memoryless sources defined on $\infty$-alphabets. 
Since the same result for lossless source coding is not possible~\cite{gyorfi_1994}, an interpretation of Theorem~\ref{th_achie_almost_lossless_universality} is that a non-zero distortion (for any finite block-length) is strictly needed to make the average redundancy of an universal coding scheme vanishing with the block-length. To obtain this result,  the two-stage approach  presented in Section~\ref{sub_sec:quantization_alphabet} was considered. 

If we restrict the family of two-stage schemes to have an exhaustive first-stage mapping, i.e., $\pi_n(x)= \left\{ x \right\}$ for all $x \in \mathcal{X}$ and $n \geq 1$,  then we reduce the approach to the lossless setting (i.e.,  zero distortion for every finite block-length).  In this case, if we apply the condition to obtain Theorem \ref{th_achie_almost_lossless_universality} (stated in Lemma~\ref{lemma_two_stage_suffcond_al_usc} in Section~\ref{proof_lemma_two_stage_suffcond_al_usc}), 
this reduces to verify 
that the i-radius of the family grows sub-linearly with the block-length (more details presented in  Sections~\ref{proof_th_achie_almost_lossless_universality}
 and \ref{proof_lemma_two_stage_suffcond_al_usc}), which is the condition known  for a family of distributions to have a nontrivial minimax redundancy rate~\cite{davisson_1973,kieffer_1978,gyorfi_1994,boucheron_2009,csiszar_2004}.
 
\subsection{Entropy Estimation with an Almost Lossless Universal Code: A Side Comment}
\label{sec_learning_analysis_of_alsc}
In the lossless case, the existence of a weak minimax source coding scheme $ \left\{f_n: \mathcal{X}^n  \rightarrow   \left\{ 0,1\right\}^*, n\geq 1 \right\}$ for a family of distribution $\Lambda$ implies that 
$\sup_{\mu\in \Lambda} \lim_{n \rightarrow \infty} \big(r(f_n, \mu^n) -H(\mu)\big)=0$ \cite{csiszar_2004}. Consequently,  the average length of the code $r(f_n, \mu^n)=\mathbb{E} \left\{ \mathcal{L}(\mathcal{C}_n(\Phi_n(X^n))) \right\}/n$ is a weak consistent estimator of the entropy distribution-free in  $\mu\in \Lambda$ \cite{beirlant_1997}.  For the family of finite entropy stationary and memoryless sources, we have that it is not feasible to have a weak minimax USC scheme in $\infty$-alphabets.   
In fact, \cite[Theorem 2]{gyorfi_1994} says that for every code $f_n$ and $n\geq 1$, there exists $\mu\in \mathcal{P}_{\mathcal{H}}(\mathcal{X})$ such that $r(f_n, \mu^n)=\infty$. In other words, there is no lossless variable-length source coding scheme that offers a weakly consistent estimator of the entropy using its average block-length (per letter).
In contrast, Theorem \ref{th_achie_almost_lossless_universality} shows that there is an almost lossless 
USC scheme $\big\{(\phi_n, \psi_n,$ $ \mathcal{C}_n, \mathcal{D}_n): n\geq 1 \big\}$ with an average length that offers a distribution-free weakly consistent estimation of the entropy in $\mathcal{P}_{\mathcal{H}}(\mathcal{X})$. In fact from the proof of Lemma \ref{lemma_two_stage_suffcond_al_usc} (Sec.\ref{proof_lemma_two_stage_suffcond_al_usc}), we have that 
\begin{equation*}
\lim_{n \rightarrow \infty} \sup_{\mu\in \mathcal{P}_\mathcal{H}(\mathcal{X})} \left(r(\phi_n, \mathcal{C}_n, \mu^n) -H_{\sigma(\pi_n)}(\mu)\right)=0,
\end{equation*}
and from the fact that $\left\{ \pi_n: n\geq 1\right\}$ is asymptotically  sufficient for $\mathcal{P}_{\mathcal{H}}(\mathcal{X})$ (Definition \ref{def_asymtotically_suff_part_family} in Section~\ref{proof_th_achie_almost_lossless_universality}), it follows that:  $\lim_{n \rightarrow \infty} r(\phi_n, \mathcal{C}_n, \mu^n) = H(\mu)$, for all $\mu\in \mathcal{P}_{\mathcal{H}}(\mathcal{X})$. Then, by relaxing the lossless block-wise assumption (introducing a  non-zero distortion), we control the worse-case redundancy, which  is bounded by the i-radius of $\mathcal{P}_{\mathcal{H}}(\mathcal{X})$ restricted to a sub-sigma field (see (\ref{eq_sub_main_2})). This flexibility enables the capacity to balance two sources of errors: 
 $r(\phi_n, \mathcal{C}_n, \mu^n) -H_{\sigma(\pi_n)}(\mu)$ (a kind of estimation error)  and $H(\mu)-H_{\sigma(\pi_n)}(\mu)$ (an approximation error),  that at the end offers an distribution-free estimate of the entropy (point-wise) using the average length of the code.

\section{Uniform Convergence of the Distortion} 
\label{sub_sec:uniform_convergence_distortion}
In this section, we further focus on a stronger notion of universal weak source coding.   We study whether is possible to achieve an  uniform convergence of the distortion to zero (over the entire family $\Lambda$), instead of the point-wise convergence stated in Definition \ref{def_universality}. To this end, we restrict the analysis to the rich family of envelope distributions introduced in Section \ref{sub_sec_envelop_results_usc}. 
We can state the following dichotomy: 
\begin{theorem}[Uniform convergence] \label{th_dichotomy_stronger_universality}
	Let us consider the family of envelope distributions $\Lambda_f$.  
	\begin{itemize}
		\item[i)] If $f\in \ell_1(\mathcal{X})$, then there is a two-stage coding scheme 
		$\left\{ (\phi_n, \psi_n,\mathcal{C}_n): n\geq 1 \right\}$ with $\left| \pi_n\right| < \infty$ (finite size) such that
		\begin{align}
		&\lim_{n \rightarrow \infty} \sup_{\mu\in \Lambda_f} d(\phi_n, \psi_n,\mu^n) = 0, \text{ and }	\nonumber\\
		&\lim_{n \rightarrow \infty} \sup_{\mu\in \Lambda_f} \Big(r(\phi_n, \mathcal{C}_n, \mu^n) -H(\mu)\Big)=0.\nonumber
		\end{align}
		
		\item[ii)] Otherwise, i.e., $f\notin \ell_1(\mathcal{X})$, for any two-stage 
		code $(\phi_n, \psi_n,\mathcal{C}_n)$ of length $n$ with $\left| \pi_n\right| < \infty$  it follows that 
		\begin{equation*}
			\sup_{\mu\in \Lambda_f} d(\phi_n, \psi_n,\mu^n) = 1,
		\end{equation*}
		while if $\left| \pi_n\right| = \infty$, then 
		\begin{equation*}
			\sup_{\mu\in \Lambda_f} \Big (r(\phi_n, \mathcal{C}_n, \mu^n) -H(\mu)\Big )= \infty.
		\end{equation*}
		More generally, for a lossy code $(f_n,g_n)$ of length $n$, provided that 
		\begin{equation*}
			\sup_{\mu \in \Lambda_f} d(f_n,g_n,\mu^n)<1, 
		\end{equation*}
		then
		\begin{equation*}
			\sup_{\mu \in \Lambda_f} \Big(r(f_n,\mu^n) -H(\mu)\Big)=\infty.\nonumber
		\end{equation*}
	\end{itemize}
\end{theorem}

The proof is presented in Section \ref{proof_th_dichotomy_stronger_universality}.

	Theorem \ref{th_dichotomy_stronger_universality} states that if the envelope function is summable, there is a two-stage coding scheme of finite size 
	that offers a uniform convergence of the distortion to zero (over $\Lambda_f$), while ensuring that the worst-case average 
	redundancy (over $\Lambda_f$) vanishes with the block-length. 
	On the negative side, for all stationary memoryless sources indexed by a non-summable envelope function, it is not possible to achieve an uniform 
	convergence of the distortion to zero with  
	a finite size two-stage coding rule. An infinite size rule is indeed needed, 
	i.e., $\left| \pi_n\right| = \infty$, eventually with $n$, that on the down-side it has an unbounded i-radius 
	(details presented in Lemmas \ref{pro_invariance_envelop} and \ref{dichotomy_envelops} at Section \ref{proof_th_dichotomy_stronger_universality}). 
	Importantly, this impossibility result remains when enriching the analysis with the adoption of general lossy coding rules (details in Sec. \ref{converse_general_lossy_scheme}).
	
	Finally, it worths noting that the family $\Lambda_f$ with $f\in \ell_1(\mathcal{X})$ has a 
	finite regret and redundancy in the context of lossless universal source coding \cite[Ths. 3 and 4]{boucheron_2009}. Furthermore,  summability is the necessary and sufficient condition on $f$ that makes this collection strongly minimax universal 
	in  lossless source coding \cite{boucheron_2009}. Then, based on this strong almost lossless source coding criterion
	(with a uniform convergence to zero of the distortion and the redundancy) 
	it is not possible to code (universally) a richer family of distributions when restricting the analysis to envelope families.

\section{Redundancy Gains for Summable Envelope Families}
\label{sec_redundancy_gains_envelop_families}

Theorem \ref{th_dichotomy_stronger_universality} states that we can achieve a uniform convergence 
of the distortion to zero while the worst-case redundancy vanishes with a two-stage lossy scheme if,  and only if, $\Lambda_f$ has a summable envelope function. On the lossless variable-length source coding side,  if $f\in \ell_1(\mathcal{X})$ we know from Theorem \ref{dichotomy_envelops} that the i-radius of the family
\begin{equation*}
R^+(\Lambda_f^n) =   \min_{v^n \in \mathcal{P}(\mathcal{X}^n)} \sup_{\mu^n \in \Lambda_f^n} \mathcal{D}(\mu^n | v^n),
\end{equation*}
is $o(n)$  \cite[Ths. 3 and 4]{boucheron_2009}, 
which is equivalent to state that $\Lambda_f$  
is strongly minimax universal \cite{csiszar_2004}. Therefore, under the assumption that  $f\in \ell_1(\mathcal{X})$, the lossy approach with asymptotic vanishing distortion may appear to not be useful if no gains are observed in the way the worst-case redundancy approaches zero in (\ref{eq_sec_background_8b}), with respect  to the normalized i-radius sequence $(R^+(\Lambda_f^n)/n)_{n\geq 1}$ that governs minimax redundancy in the lossless case \cite{csiszar_2004,boucheron_2009}.

This section explores the feasibility of obtaining gains in terms of the minimax redundancy of a  two-stage lossy approach tends to zero, when compared to the minimax redundancy (of the lossless scenario) for $\Lambda_f$ when $f\in \ell_1(\mathcal{X})$. We focus on the finite size {\em tail-based partition} scheme used to prove Theorem \ref{th_achie_almost_lossless_universality} and the achievability part of Theorem \ref{th_dichotomy_stronger_universality}.

\subsection{Preliminaries}
Let us consider a positive and non-decreasing sequence of integers $(k_n)_{n}$ and the collection of {\em tail partitions} induced by: 
\begin{equation}\label{eq_sec_rgef_1}
	\tilde{\pi}_{k_n} \coloneqq    \left\{  \left\{ 1\right\}, \cdots,  \left\{ k_n-1\right\}, \Gamma_{k_n-1}^c \right\}, \ \forall n\geq 1, 
\end{equation}
where $\Gamma_{k}= \left\{1,\dots ,k\right\}$. Note that $\tilde{\pi}_{k_n}$ resolves all the elements of $\Gamma_{k_n-1}= \left\{1,\dots,k_n-1 \right\}$ and consequently,  there is a pair $(\tilde{\phi}_n, \tilde{\psi}_n)$ associated with $\tilde{\pi}_{k_n}$ such that $\forall \mu \in \Lambda_f$, 
\begin{IEEEeqnarray}{rCl}
	d(\tilde{\phi}_n, \tilde{\psi}_n,\mu^n) &\leq& \mu(\Gamma_{k_n-1}^c) 
\leq	\sum_{x\geq k_n} f(x). \label{eq_sec_rgef_2}
	\end{IEEEeqnarray}
Consequently, it follows that 
\begin{equation}
\sup_{\mu \in \Lambda_f} d(\tilde{\phi}_n, \tilde{\psi}_n,\mu^n) \leq \sum\limits_{x \geq k_n} f(x)<\infty. 
\end{equation}
It is then easy to verify that $(1/k_n)_n$ being  $o(1)$ is the necessary and sufficient condition for the tail-based scheme to have the uniform convergence (over $\Lambda_f$)  of the distortion to zero. 

Concerning the worst-case minimax redundancy of the two-stage scheme induced by $\left\{ \tilde{\pi}_{k_n}\right\}$ in (\ref{eq_sec_background_8b}), 
we can consider a lossy mapping  $(\tilde{\phi}_n, \tilde{\psi}_n)$ consistent with $\tilde{\pi}_{k_n}$ (first-stage), where it is clear 
that the entropy of $Y^n=(\tilde{\phi}_n (X_1),..,\tilde{\phi}_n(X_n))$, which is $H(\tilde{\Phi}_n(X^n))=n H_{\sigma(\tilde{\pi}_{k_n})}(\mu) \leq n H(\mu)$,  
is a  
lower bound for the performance of any prefix-free code acting on  $Y^n$. Then given the first-stage $(\tilde{\phi}_n, \tilde{\psi}_n)$,   
we define the worse-case redundancy of  any $\mathcal{C}_n:\Gamma_{k_n}^n \rightarrow \left\{0,1\right\}^*$ as follows: 
\begin{equation*}
	\bar{R}(\Lambda_f^n,  \tilde{\pi}_{k_n}, \mathcal{C}_n) 
	\coloneqq    \sup_{\mu \in \Lambda_f} \big(r(\tilde{\phi}_n,\mathcal{C}_n, \mu^n ) - H_{\sigma(\tilde{\pi}_{k_n})}(\mu)\big).
\end{equation*}
Therefore for any finite $n$ and first-stage 
partition $\tilde{\pi}_{k_n}$,  the minimax redundancy of the second-stage is: 
\begin{equation} \label{eq_sec_rgef_1b}
	\min_{\mathcal{C}_n: \Gamma^n_{k_n} \rightarrow \left\{0,1\right\}^*} \bar{R}(\Lambda_f^n,  \tilde{\pi}_{k_n}, \mathcal{C}_n). 
\end{equation}
For the second stage, we can use again the connection between prefix-free codes and distributions to map  $\mathcal{C}_n$ to a probability $v$ in 
$\mathcal{P}(\Gamma_{k_n}^n)$,  where the redundancy $r(\tilde{\phi}_n,\mathcal{C}_n, \mu^n ) - H_{\sigma(\tilde{\pi}_{k_n})}(\mu)$ can be expressed as 
one over $n$ the divergence restricted to the cells of $\tilde{\pi}_{k_n}$ \cite{kullback1958}, more precisely as $\frac{1}{n } \mathcal{D}_{\sigma(\tilde{\pi}_{k_n} \times\cdots\times \tilde{\pi}_{k_n})}(\mu^n | v) $ where: 
\begin{equation*} 
\mathcal{D}_{\sigma(\tilde{\pi}_{k_n} \times\cdots\times \tilde{\pi}_{k_n})}(\mu^n | v) \coloneqq  \sum_{\mathcal{A}\in \tilde{\pi}_{k_n} \times\cdots\times \tilde{\pi}_{k_n}} \mu^n(\mathcal{A})\log \frac{\mu^n(\mathcal{A})}{v(\mathcal{A})},
\end{equation*}
with $\tilde{\pi}_{k_n} \times\cdots\times \tilde{\pi}_{k_n}$ being a short-hand for the product partition of $\mathcal{X}^n$ induced by $\tilde{\pi}_{k_n}$.
Consequently, the USC problem of the second-stage in (\ref{eq_sec_rgef_1b}) given the first stage, i.e., given $\tilde{\pi}_{k_n}$ in (\ref{eq_sec_rgef_1}), 
can be expressed by $\frac{1}{n} R^+(\Lambda_f^n, \sigma(\tilde{\pi}_{k_n}))$ where:
\begin{equation}\label{eq_sub_main_1c}
	R^+(\Lambda_f^n, \sigma(\tilde{\pi}_{k_n})) \coloneqq  \min_{v \in \mathcal{P}(\mathcal{X}^n)} \sup_{\mu^n \in \Lambda_f^n} \mathcal{D}_{\sigma(\tilde{\pi}_{k_n} \times\cdots\times \tilde{\pi}_{k_n})}(\mu^n | v). 
\end{equation}
$R^+(\Lambda_f^n, \sigma(\tilde{\pi}_{k_n}))$ can be interpreted as the i-radius of $\Lambda_f^n$ restricted to the events of the sub-sigma field $\sigma(\tilde{\pi}_{k_n} \times\cdots\times \tilde{\pi}_{k_n})$.\footnote{More details are presented in (\ref{eq_lemma_two_stage_4}) and (\ref{eq_lemma_two_stage_5}) in Sec. \ref{proof_lemma_two_stage_suffcond_al_usc}.} 

\subsection{Redundancy Gain Analysis}
\label{sub_sec_redundancy_result}
Returning to our question, in the context of Eq. (\ref{eq_sub_main_1c}) we know that $\mathcal{D}_{\sigma(\tilde{\pi}_{k_n} \times\cdots\times \tilde{\pi}_{k_n})}(\mu^n|v) \leq \mathcal{D}(\mu^n|v)$ \cite{kullback1958}, therefore  for any sequence $(k_n)_{n}$ of positive integers, it follows that $R^+(\Lambda_f^n)\geq R^+(\Lambda_f^n, \sigma(\tilde{\pi}_{k_n}))$ for all $n$,  Consequently, we have that  
\begin{equation}
\liminf_{n \rightarrow \infty} \frac{R^+(\Lambda_f^n)}{R^+(\Lambda_f^n, \sigma(\tilde{\pi}_{k_n}))}\geq 1.
\end{equation}
In particular, we want to determine regimes on $(k_n)_{n}$ that guarantee an asymptotic gain in minimax redundancy in the sense that
\begin{equation}\label{eq_sec_rgef_3}
\lim_{n \rightarrow \infty} \frac{R^+(\Lambda_f^n, \sigma(\tilde{\pi}_{k_n}))}{R^+(\Lambda_f^n)} =0,
\end{equation}
subject to the condition that $(1/k_n)_n$ is  $o(1)$.  If $(k_n)_{n}$ offers an asymptotic gain on minimax redundancy in the sense of Eq. (\ref{eq_sec_rgef_3}), then any sequence where $(\tilde{k}_n)_n$ such that 
$(\tilde{k}_n)_n \leq (k_n)_n$  eventually in $n$  
offers a gain in the minimax redundancy\footnote{This is a simple consequence of the fact that $k\geq \tilde{k}$ implies $\sigma(\tilde{\pi}_{\tilde{k}}) \subset \sigma(\tilde{\pi}_k)$.}. Therefore it is important to determine  the largest size sequence  for the tail partition such that  (\ref{eq_sec_rgef_3}) is satisfied.

Note that any sequence $(k_n)_{n}$ offers a non-zero worst-case distortion for a finite block-length,  and if $(1/k_n)_n$ is $o(1)$ this worst case distortion goes to zero at rate function of $(k_n)_n$ and  the envelope function $f$. From this, one could suspect a gain in the minimax redundancy,  in the sense established by (\ref{eq_sec_rgef_3}), no matter how fast $(k_n)_n$ tends to infinity with the block-length as long as $k_n<\infty$ for any $n$.  In other words, one simple conjecture is that the complexity of a family of distributions with infinite degrees of freedom, measured in terms of the rate of convergence to zero of the minimax redundancy per sample $({R^+(\Lambda_f^n)}/{n})_{n}$, cannot be reached by projecting this family into  finite but dynamic (with the block-length) alphabets. However, the following result refutes this initial guess and determines a non-trivial regime for $(k_n)_{n}$ with no minimax redundancy gain. Importantly, this regime is fully determined by $(u^*_f(n))_n$ (see Def.\ref{def_tail_function_critical_dimension_b} in Section \ref{sub_sec_envelop_results_usc}), which can be interpreted as a sequence of critical dimensions for $\Lambda_f$ that was introduced by Bontemps {\em et al.} \cite{bontemps_2011, bontemps_2014} in the context of lossless USC.

\begin{theorem}[Minimax redundancy gains] \label{th_information_radius_gain_tail_schemme}
	Let $\Lambda_f \subset \mathcal{P}(\mathcal{X})$ be an \emph{envelope family} with $f\in \ell_1(\mathcal{X})$ and  $\left|\textrm{supp}(f)\right|=\infty$. 
	In addition, let $  \left\{ \tilde{\pi}_{k_n}: n\geq 1 \right\} $  be the collection of tail-based partitions  in (\ref{eq_sec_rgef_1}) 
	driven by a positive non-decreasing sequence $(k_n)_{n}$. It follows that:
	\begin{itemize}
		\item[i)] If $(k_n)_n \geq (u^*_f(n))_n$ eventually in $n$, then there is no gain in minimax redundancy in the sense that: 
		\begin{equation*}
			\lim_{n \rightarrow \infty} \frac{R^+(\Lambda_f^n, \sigma(\tilde{\pi}_{k_n}))}{R^+(\Lambda_f^n)} =1.
		\end{equation*}
		\item[ii)] Conversely, if $(k_n)_n$ is $o(u^*_f(n))$, i.e., $\lim_{n \rightarrow \infty }{k_n}/{u^*_f(n)}=0$, 
		then we have a minimax redundancy gain:
		\begin{equation*}
			\lim_{n \rightarrow \infty} \frac{R^+(\Lambda_f^n, \sigma(\tilde{\pi}_{k_n}))}{R^+(\Lambda_f^n)} =0.
		\end{equation*}
	\end{itemize}	
\end{theorem}
The proof is presented in Section \ref{proof_th_information_radius_gain_tail_schemme}.

Analysis and interpretation of Theorem \ref{th_information_radius_gain_tail_schemme}:
\begin{enumerate}
	\item First, we note that the sequence $(u^*_f(n))_n$ in (\ref{eq_sec_rgef_4}) defines a notion of critical dimension (or cardinality) for the family $\Lambda_f$,  as it characterizes a boundary (or phase transition) on the size of the tail-based two-stage coding schemes above from which no gains in terms of the rate of minimax redundancy are obtained.
	
	\item If we consider the regime of redundancy gain, i.e., where $(k_n)_n$ is  $o(u^*_f(n))$,   it is simple to note that any arbitrary partition 
	scheme $\left\{\pi_n: n\geq 1 \right\}$ such that $\left| \pi_n \right|= \left| \tilde{\pi}_{k_n}\right| =k_n$ satisfies:
	$$\lim_{n \rightarrow \infty} \frac{R^+(\Lambda_f^n, \sigma({\pi}_{n}))}{R^+(\Lambda_f^n)} =0.$$
	Then, this scenario of redundancy gain can be extended to any finite alphabet partition strategy, and consequently, we can say that 
	the condition $(k_n)_n$ is $o(u^*_f(n))$ offers a {\em trivial}  regime of minimax redundancy gain. However, what is not evident is the fact that the tail-based partition offers a non-trivial regime of redundancy gain, in the sense that the condition $(k_n)_n \geq (u^*_f(n))_n$ eventually with $n$ suffices  to guarantee that: 
	$$\lim_{n \rightarrow \infty} \frac{R^+(\Lambda_f^n, \sigma(\tilde{\pi}_{k_n}))}{R^+(\Lambda_f^n)} =1.$$
	From this angle, the tail-based partition is efficient (or sufficient) to capture the asymptotic complexity of $\Lambda_f$ with a minimum alphabet size. Complementing this richness property of $\left\{ \tilde{\pi}_{k_n} \right\}$, it is simple to verify that the  tail partition is an optimal solution when the objective is to minimize the worst-case distortion of a two-stage lossy coding scheme restricting the finite size $k>0$ on the quantization.
	
	\item From a complexity view-point, $R^+(\Lambda_f^n)$ measures the complexity of the lossless coding task. 
	Then, for a given finite partition $\pi_n \in \mathcal{B}(\mathcal{X})$,  $R^+(\Lambda_f^n)- R^+(\Lambda_f^n, \sigma(\pi_{n})) \geq 0$ can be interpreted as the reduction on complexity by the process of projecting $\Lambda^n_f$ into a finite alphabet, i.e.,  
	$$\Lambda_f^n/\sigma({\pi}_{n}) \coloneqq    \big\{\mu^n/\sigma{({\pi}_{n} \times \cdots \times {\pi}_{n})}:\mu \in \Lambda_f \big\}$$
	where $\mu/\sigma{(\pi)} \coloneqq \left\{\mu(\mathcal{A}): \mathcal{A}\in \sigma{(\pi)} \right\} \in \mathcal{P}(\mathcal{X}, \sigma(\pi))$ is a short-hand for the probability $\mu$  restricted to the sub-sigma field induced by $\pi$ (details presented in Section \ref{proof_subsec_no_gain_i_radius}).
	Then, it is interesting to know if the i.i.d. family of envelope distributions $ \left\{ \Lambda^n_f: n\geq 1\right\}$ with $f\in \ell_1(\mathcal{X})$ 
	admits a finite but dynamic alphabet reduction that captures its complexity asymptotically with $n$. For that question, we can introduce the following: 
	\begin{definition}
	\label{def_sufficient_size_lambda}
	We say that $\Lambda \subset \mathcal{P}(\mathcal{X})$ has a finite alphabet reduction, if there exists a partition scheme $\left\{ \pi_n: n\geq 1 \right\} $ with $  \left| \pi_n \right|=k_n<\infty$ such that $\lim_{n \rightarrow \infty} \frac{R^+(\Lambda^n, \sigma({\pi}_{n}))}{R^+(\Lambda^n)} =1$, or, equivalently, that $\left\{ \Lambda^n: n\geq 1\right\}$ is equivalent to $ \left\{ \Lambda^n/\sigma(\pi_n): n\geq 1 \right\}$  in terms of asymptotic information complexity. In this case, we  say that $(k_n)_n$ 
	is a sequence  of {\em sufficient sizes} (or sufficient) to represent $\Lambda$.
	\end{definition}
	\begin{definition} 
	We say that  $(k^*_n)_n$ is the {\em critical (or minimal) size to represent} $\Lambda_ f$, if  $(k^*_n)_n$ is a  sequence of sufficient size to represent $\Lambda$ (Def. \ref{def_sufficient_size_lambda}), and no sequence $(l_n)_n$ exists such that: $(l_n)_n$ is sufficient to represent $\Lambda$ and $(l_n)_n \ll (k^*_n)_n$. 
	\end{definition}
	
	In this context,  the proof of Theorem \ref{th_information_radius_gain_tail_schemme} shows as a corollary that $(u^*_f(n))_{n}$ is the critical size to represent $\Lambda_ f$.  The achievability part is obtained using the tail-based partition and some metric entropy lower bound  
	for the i-radius extended from \cite{haussler_1997,bontemps_2011,bontemps_2014}.  On the other hand, the converse argument derives from  basic i-radius results for i.i.d. sources over finite alphabets \cite{csiszar_2004} and results for envelope families on countably infinite alphabets \cite{boucheron_2009}.
	
	Finally, from  Boucheron {\em et al.} \cite[Ths. 3 and 4, and Cor. 2]{boucheron_2009},  
	we have that $\Lambda_f$ has either a finite alphabet reduction with a sub-linear critical size sequence given by $(u^*_f(n))_{n}$ (if  $f\in \ell_1(\mathcal{X})$), or infinite minimax redudancy for all $n\geq 1$. 
	
	\item To illustrate the result, let us consider the exponentially decreasing envelope class studied in \cite{bontemps_2011,boucheron_2009}: 
	$$\Lambda_{f_{Ce^{-\alpha}}}=\left\{\mu\in \mathcal{P}(\mathcal{X}), f_\mu(x) \leq f_{Ce^{-\alpha}}(x)= C e^{-\alpha x} \right\},$$
	where $C>0$ and $\alpha>0$. It has been shown in \cite[Prop. 6]{bontemps_2011} that
	$$\frac{1}{\alpha} \ln(C x) \leq U_{f_{Ce^{-\alpha}}}(x) +1 \leq \frac{1}{\alpha} \ln(\kappa C x)$$
	where $\kappa = 1/(1-e^{-\alpha})$ and $U_{f}(x)$ is defined in Def. \ref{def_continuos_Uf} (see Section \ref{proof_th_information_radius_gain_tail_schemme} for details). 
	Importantly for our analysis, it follows that  $u^*_{f_{Ce^{-\alpha}}}(n)-1 \leq U_{f_{Ce^{-\alpha}}}(n) <  u^*_{f_{Ce^{-\alpha}}}(n)$ (see  Section \ref{proof_th_information_radius_gain_tail_schemme_gain_regime}), therefore we have that for all $n\geq 1$
	$$\frac{1}{\alpha} \ln(C n) +1 <  u^*_{f_{Ce^{-\alpha}}}(n) \leq \frac{1}{\alpha} (\ln (\kappa) +\ln(C n) ).$$
	Consequently,  the critical dimension of this exponential family, which determines the regime of redundancy gain, 
	scales like $\sim (\frac{1}{\alpha}  \ln n)_n$. Similar analysis can be conducted on the power-law envelopes and sub-exponential envelopes  classes
	studied in \cite{bontemps_2014,bontemps_2011,boucheron_2009,acharya_2014}. See also an excellent exposition of these last results in \cite{gassiat_2018}.
\end{enumerate}

\section{Summary and Concluding Remarks} 
\label{sec_discusion}
The problem of almost lossless universal source coding for countably infinite alphabet sources is introduced in this work. 
Our main result  shows that a weak notion of universal (variable length) source coding is feasible for the entire 
class of finite entropy stationary memoryless sources. This result is obtained by  tolerating a (non-zero) single-letter distortion in the encoding process 
that vanishes asymptotically with the block-length. To this end, one key idea is an induced sequence of partitions of the  $\infty$-alphabet, which offers a way to control the  worst-case average redundancy associated with the i-radius of a class of distributions restricted to a subsigma-field. 
We have also studied a stronger almost losses condition, asking for uniform convergence of the distortion to zero (over the family of distributions), where it turns out that this variation of weak universality can be achieved  for the same class of envelope distributions that is strong minimax universal in the lossless case.  This last result suggests that asking for a non-zero distortion that convergence to zero point-wise (over the family of distributions) is the strongest relaxation from the lossless criterion 
that allows us to control the  worst-case redundancy of the problem. 
Finally,  we show that it is possible  to obtain gains in the rate of convergence of the worst-case redundancy  of an almost lossless scheme,  with respect to the worst-case redundancy of the lossless case, by tolerating a non-zero distortion that tends to zero with the block-length. In this context,  we fully characterize the regime of  gains for a two-stage lossy scheme induced by tail based partitions. 
\section{Proofs of the Main Results} 
\label{section_proof}

\subsection{Theorem \ref{theorem_entropy}}
\label{proof_theorem_entropy}
First, we introduce a 
result and some definitions that will be used in the proof. 
\subsubsection{Preliminaries}
\begin{definition} \label{def_rate_distortion}
For 
$\mu\in \mathcal{P}(\mathcal{X})$ its rate distortion function  is given by \cite{berger_1998,cover_2006,gray_1990}: 
\begin{equation*} 
	R_{\mu}(d)\coloneqq    \inf_{P(\tilde{X}|X)\, \textrm{ st. } \, \mathbb{P}(\tilde{X} \neq X) \leq d} I(X; \tilde{X}). 
\end{equation*}
\end{definition}

\begin{definition}\label{def_truncated_distribution}
For any $\mu\in  \mathcal{P}(\mathcal{X})$ and $\theta>0$, let us define $\tilde{\mu}_\theta \in  \mathcal{P}(\mathcal{X})$
by: $\tilde{\mu}_{\theta}(\left\{i \right\}) \coloneqq  \min \left\{\theta, \mu(\left\{i \right\}) \right\}$ for all $i>1$ and 
$\tilde{\mu}_{\theta}(\left\{1 \right\}) \coloneqq  1-\kappa_{\theta}$ where  $\kappa_{\theta} \coloneqq \sum_{i>1} \tilde{\mu}_{\theta}(\left\{i \right\})$. 
\end{definition}

\begin{lemma} (Ho {\em et al.}\cite[Th. 1]{ho_2010}) 
\label{lemma_asymt_rate_dist_func_countable}
Let us consider $\mu\in \mathcal{P}(\mathcal{X})$, where $\mathcal{X}$ is an $\infty$-alphabet, then there is $d_0>0$ such that $\forall d \leq d_0$ 
\begin{equation} \label{eq_proof_theorem_entropy_4}
	R_{\mu}(d) = H(\mu) -  H(\tilde{\mu}_{\theta(d)})
\end{equation}
with $\tilde{\mu}_{\theta}$ introduced in Def.\ref{def_truncated_distribution} 
and $\theta(d) > 0$ being the solution of the condition $\kappa_{\theta} = \sum_{i>1} \tilde{\mu}_{\theta}(\left\{i \right\})=d$. 
\end{lemma}

\subsubsection{Proof of Theorem \ref{theorem_entropy}}
We consider the non-trivial case where $\mu$ has infinite support over $\mathcal{\mathcal{X}}$, i.e., $\inf_{x\in \mathcal{X}} f_\mu(x)=0$,  otherwise the problem reduces to a finite alphabet scenario where this result is known \cite{cover_2006,csiszar_2004}. 

We begin with the converse argument. This reduces to prove that any lossy coding scheme with zero asymptotic distortion 
has a rate that convergences to a limit that is grater or equal to $H(\mu)$ (see Def. \ref{def_achievable}). 
For that, let us assume that we have a lossy scheme $\left\{ (f_n, g_n): n\geq 1\right\}$ such that
\begin{align} \label{eq_proof_theorem_entropy_1}
	&\lim_{n \longrightarrow \infty} d(f_n, g_n,\mu^n) =0 \Leftrightarrow \nonumber\\
	 &\lim_{n \longrightarrow \infty}  \frac{1}{n} \sum_{i=1}^n \mathbb{P}  \big\{  X_i\neq (g_n(f_n(X^n)))_i \big\} =0.
\end{align}
If we denote by $\hat{X}^n \coloneqq g_n(f_n(X^n))$ the reconstruction,  from lossless variable length source coding  it is well-known that \cite{csiszar_2004}:
\begin{IEEEeqnarray}{rCl} 	
	r(f_n, \mu^n) &\geq&  \frac{1}{n} I(X^n, \hat{X}^n) \nonumber\\
			\label{eq_proof_theorem_entropy_2b}	
			    &\geq &\frac{1}{n} \sum_{i=1}^n I(X_i;\hat{X}_i) 
			     \geq \frac{1}{n} \sum_{i=1}^n R_{\mu}\big(\mathbb{P}\{X_i\neq \hat{X}_i\}\big)
			    \nonumber\\
			    &\geq & R_{\mu}  \left( \frac{1}{n} \sum_{i=1}^n \mathbb{P}  \big\{  X_i\neq  \hat{X}_i \big\}\right) \nonumber\\
			    &=& R_{\mu} \big(d(f_n, g_n,\mu^n)\big), 
\end{IEEEeqnarray}
where for the inequalities in (\ref{eq_proof_theorem_entropy_2b}), we use that $\mathbf{X}$ is memoryless, the non-negativity of the conditional mutual information \cite{cover_2006}, and the convexity of the rate-distortion function of $\mu$  \cite{berger_1998,cover_2006}.  
For the rest we assume that $\mu$ is organized  in decreasing order in the sense that $f_\mu(1) \geq f_{\mu}(2)\geq \cdots$ \footnote{We note that  for the
charcaterization of $R_{\mu}(d)$ as $d \rightarrow 0$ this assumption implies no loss of generality.} and that we are in the regime where $d\leq d_o$ (introduced in Lemma \ref{lemma_asymt_rate_dist_func_countable}). 
Using Lemma \ref{lemma_asymt_rate_dist_func_countable}, we have that $R_{\mu}(d) = H(\mu) -  H(\tilde{\mu}_{\theta(d)})$,  where if we consider
\begin{equation} \label{eq_proof_theorem_entropy_5}
	K_\mu(d) \coloneqq     \min \big\{ k>1:
								f_\mu(k+1) \leq \theta(d) \big\}, 
\end{equation}
it is simple to verify that:
\begin{align} \label{eq_proof_theorem_entropy_6}
	&H(\tilde{\mu}_{\theta(d)}) = (1-\kappa_{\theta(d)})\log \frac{1}{1-\kappa_{\theta(d)}} \nonumber\\
	&+ (K_\mu(d)-1)\cdot \theta(d) \log \frac{1}{\theta(d)} + \sum_{i>K_\mu(d)} f_\mu(i) \log \frac{1}{f_\mu(i)}. 
\end{align}
From (\ref{eq_proof_theorem_entropy_4}) and (\ref{eq_proof_theorem_entropy_1}), we focus on exploring $H(\tilde{\mu}_{\theta(d_n)})$ when $d_n \rightarrow 0$.  First,  it is simple to verify that $d \rightarrow 0$ implies that $\theta(d) \rightarrow 0$ by definition.  Then, for a fix $\mu \in \mathcal{P}_{\mathcal{H}}(\mathcal{X})$ with infinite support, the problem reduces to chacaterize $\lim_{\theta \rightarrow 0} H(\tilde{\mu}_{\theta})$.  Note that $\tilde{\mu}_{\theta}$ converges point-wise to the degenerate probability $\mu^*=(1,0,\cdots)$ as $\theta$ vanishes~\footnote{In the  $\infty$-alphabet 
the point-wise convergence of probabilities to a limit is equivalent to the weak convergence and the convergence in total variations \cite{piera_2009}.}. 
However, by the entropy discontinuity~\cite{ho_2010b, silva_2012_jspi,silva_2012_isit}, the convergence of the measure to $\mu^*$ is not sufficient to guarantee  that  $\lim_{\theta \rightarrow 0} H(\tilde{\mu}_{\theta})=H(\mu^*)=0$.

First, it is simple to note that $\kappa_\theta \rightarrow 0$, as $\theta \rightarrow 0$ considering that $\tilde{\mu}_{\theta}(i) \rightarrow 0$ for all $i\geq 1$,  $\mu\in \mathcal{P}_\mathcal{H}(\mathcal{X})$, and 
the {\em dominated convergence theorem} \cite{varadhan_2001}. 
Then, $\lim_{\theta \rightarrow 0} (1-\kappa_\theta)\log \frac{1}{1-\kappa_\theta}=0$, which is the limit of the first term in the RHS of (\ref{eq_proof_theorem_entropy_6}). For the rest, we define the self-information function $i_{\tilde{\mu}_\theta}(i) \coloneqq    \tilde{\mu}_{\theta}(\left\{i \right\}) \log 1/ \tilde{\mu}_{\theta}(\left\{i \right\})>0$, for all $i>1$, and $i_{\tilde{\mu}_\theta}(1)\coloneqq  0$. By definition $\lim_{\theta \rightarrow 0} i_{\tilde{\mu}_\theta}(i)=0$ point-wise in $\mathcal{X}$, noting that $\lim_{\theta \rightarrow 0} K_\mu(\theta) = \infty$ and $H(\mu)<\infty$. Furthermore, there is $\theta_0>0$ such that for all  $\theta< \theta_0$, $0\leq i_{\tilde{\mu_\theta}}(i) \leq i_{\tilde{\mu}_{\theta_0}}(i)$ for all \footnote{This follows from the fact that the function $\theta \log \frac{1}{\theta}$ is monotonically increasing in the range of $(0,\theta_0)$ for some $\theta_0>0$.} $i\geq 1$, where  from the  assumption that $H(\mu)<\infty$, and the fact that $K_{\mu}(\theta_0)<\infty$, then  $ (i_{\tilde{\mu_{\theta_0}}}(i)) \in \ell_1(\mathcal{X})$.  Again by the {\em dominated convergence theorem} \cite{varadhan_2001},  $\lim_{\theta \rightarrow 0} \sum_{i\geq 1} i_{\tilde{\mu_\theta}}(i) =0$ and consequently, $\lim_{\theta \rightarrow 0} H(\tilde{\mu}_{\theta})=0$ from (\ref{eq_proof_theorem_entropy_6}).

Returning to (\ref{eq_proof_theorem_entropy_2b}), it follows that for all $n\geq 1$,
\begin{IEEEeqnarray}{rCl} \label{eq_proof_theorem_entropy_7}
	r(f_n, \mu^n) &\geq& R_{\mu} \big(\underbrace{d(f_n, g_n,\mu^n)}_{d_n \coloneqq}\big) \nonumber\\
	&=&H(\mu) - H\left(\tilde{\mu}_{\theta(d_n)}\right). 
\end{IEEEeqnarray}
Finally, as $d_n \rightarrow 0$, 
\begin{IEEEeqnarray}{rCl} 
	\liminf_{n \rightarrow 0} r(f_n, \mu^n) &\geq& H(\mu) - \limsup_{d_n \rightarrow 0 }H\left(\tilde{\mu}_{\theta(d_n)}\right) \nonumber\\
	&=& H(\mu) - \lim_{\theta \rightarrow 0} H(\tilde{\mu}_{\theta})\nonumber\\
	&=&H(\mu)  \label{eq_proof_theorem_entropy_8}. 
\end{IEEEeqnarray}
Therefore, the inequality in (\ref{eq_proof_theorem_entropy_8}) implies that $R_{al}(\mu)\geq H(\mu)$ from Definition \ref{def_achievable}.  

The achievability part (i.e., $R_{al}(\mu)\leq H(\mu)$) follows from the proof of Proposition \ref{proposition_entropy_two_stage} in Appendix~\ref{proof_proposition_entropy_two_stage}.
\hspace{\fill}~\QED

\subsection{Theorem \ref{th_achie_almost_lossless_universality}}
\label{proof_th_achie_almost_lossless_universality}
For the proof of Theorem \ref{th_achie_almost_lossless_universality}, we  first introduce some definitions and an achievability result:
\subsubsection{Preliminaries}
Regarding the distortion,  we need the following definition: 
\begin{definition}\label{def_asymtotically_suff_part_family}
	A sequence of partitions $\left\{ \pi_n: n \geq 1\right\}$ of $\mathcal{X}$ is 
	asymptotically sufficient for $\Lambda \subset \mathcal{P}(\mathcal{X})$, if  it is 
	asymptotically sufficient for every measure $\mu\in \Lambda$ (cf. Definition~\ref{def_asymtotically_suff_part}). 
\end{definition}

Concerning the analysis of the worst-case average redundancy in a lossy context,  it is instrumental to introduce the divergence restricted to a sub-sigma field \cite{kullback1958}. 
\begin{definition}\label{def_divergence_sub_sigma}
Let $\pi$ be a  partition of $\mathcal{X}$ and $\sigma(\pi) \subset \mathcal{B}(\mathcal{X})$ its induces sigma-field. Then, for every $\mu,  v \in \mathcal{P}(\mathcal{X})$,  the divergence of $\mu$ with respect to $v$ restricted to $\sigma(\pi)$ is \cite{kullback1958}:
\begin{equation}\label{eq_sub_main_1}
	\mathcal{D}_{\sigma(\pi)}(\mu|v) \coloneqq    \sum_{\mathcal{A}\in \pi} \mu(\mathcal{A})\log \frac{\mu(\mathcal{A})}{v(\mathcal{A})}.
\end{equation}
\end{definition}

\begin{definition}\label{def_i-radius_sub_sigma}
Let $\Lambda \subset \mathcal{P}(\mathcal{X})$ and $\pi$ be a partition of $\mathcal{X}$. 
Fot any $n\geq 1$, the information radius of $\Lambda^n \subset \mathcal{P}(\mathcal{X}^n)$ 
restricted to 
$\sigma(\pi)$ is given by: 
\begin{equation}\label{eq_sub_main_2}
	R^+(\Lambda^n, \sigma(\pi)) \coloneqq    \min_{v^n \in \mathcal{P}(\mathcal{X}^n)} \sup_{\mu^n \in \Lambda^n} \mathcal{D}_{\sigma(\pi \times\cdots\times \pi)}(\mu^n | v^n), 
\end{equation}
 where $\pi \times\cdots\times \pi$ denotes the product partition of $\mathcal{X}^n$, $\mathcal{P}(\mathcal{X}^n)$ is the set of probability measures in $(\mathcal{X}^n, \mathcal{B}(\mathcal{X}^n))$,  and $\Lambda^n$ denotes the collection of all i.i.d (product) probabilities measures in $(\mathcal{X}^n, \mathcal{B}(\mathcal{X}^n))$ induced by $\Lambda$. 
 \end{definition}

\begin{lemma}\label{lemma_two_stage_suffcond_al_usc}
	Let us consider $\Lambda \subset \mathcal{P}(\mathcal{X})$. 
	If there is a sequence of partitions $\left\{ \pi_n: n \geq 1\right\}$ of $\mathcal{X}$ such that:
	\begin{itemize}
	\item  $\left\{ \pi_n: n \geq 1\right\}$ is asymptotically sufficient  for $\Lambda$ (Def. \ref{def_asymtotically_suff_part_family}),  and 
	\item $(R^+(\Lambda^n, \sigma(\pi_n)))_{n}$ is $o(n)$, 
	\end{itemize}
	then the family of stationary and memoryless sources with marginal distribution in $\Lambda$ admits 
	an almost lossless source coding scheme. 
\end{lemma}

The proof 
is presented in Section~\ref{proof_lemma_two_stage_suffcond_al_usc}.

\subsubsection{Proof of Theorem \ref{th_achie_almost_lossless_universality}}
Let us consider a collection of finite size partitions $\left\{ \pi_n: n \geq 1\right\} \subset \mathcal{B}(\mathcal{X})$ 
with  $k_n=\left| \pi_n\right|<\infty$ for all $n$.  
We note that if 
\begin{equation} \label{eq_sub_main_3}
\bigcap\limits_{m\geq 1}\bigcup\limits_{l\geq m} \pi_l(x)= \left\{ x\right\}, \text{ for all }  x\in \mathcal{X} 
\end{equation}
then, this partition scheme is asymptotically sufficient for $\mathcal{P}(\mathcal{X})$.  
Concerning the information radius,  we have that $k_n = \left| \pi_n\right| <\infty$, which reduces the analysis to the finite alphabet case. In this context, it is well-known that \cite[Theorem 7.5]{csiszar_2004}:
\begin{equation}\label{eq_cor_universality_1}
	  \frac{k_n-1}{2}  \log n -K_1 \leq R^+(\mathcal{P}(\mathcal{X}^n), \sigma(\pi_n)) \leq   \frac{k_n-1}{2}  \log n + K_2,
\end{equation}
for some universal constants $K_1$ and $K_2$. Then, provided that  $(k_n)_{n}$ is  $o(n/\log n)$  it follows that  $(R^+(\mathcal{P}(\mathcal{X}^n), \sigma(\pi_n)))_n$ is  $o(n)$. 
There are numerous finite partition sequences that satisfy the  conditions stated in (\ref{eq_sub_main_3}) and $(k_n)_n$ being $o(n/\log n)$.  For example, the tail partition family given by $\bar{\pi}_{k_n} \coloneqq    \left\{ \left\{1 \right\}, \left\{2 \right\}, \cdots, \left\{ k_n-1\right\}, \Gamma_{k_n-1}^c \right\}$, where $\Gamma_{k} \coloneqq    \left\{1,\cdots,k\right\}$ and $\Gamma_{0} \coloneqq    \emptyset$, considering that  $(1/k_n)_n$ is $o(1)$ and $(k_n)_n$ is $o(n/\log n)$.  
Finally, for all  $\Lambda^n \subset \mathcal{P}(\mathcal{X}^n)$   we have by definition that $R^+(\Lambda^n, \sigma(\pi_n))\leq R^+(\mathcal{P}(\mathcal{X}^n), \sigma(\pi_n))$, which proves the result by applying Lemma~\ref{lemma_two_stage_suffcond_al_usc}. 
\hspace{\fill}~\QED

\subsection{Proof of Lemma \ref{lemma_two_stage_suffcond_al_usc}}
\label{proof_lemma_two_stage_suffcond_al_usc}
\begin{proof}
First note that if $\left\{ \pi_n: n \geq 1\right\}$ is asymptotically sufficient for the family $\Lambda$, 
it means that for all $\mu\in \Lambda$, $\lim_{n \rightarrow \infty} \mu(\cup_{k\geq n} \pi_k(x) \setminus \left\{ x\right\})=0$ (Def. \ref{def_asymtotically_suff_part}).  If we denote  by $k_n=\left| \pi_n \right|$ and $\pi_n=\left\{ \mathcal{A}_{n.i}: 1\leq i \leq k_n \right\}$, then we can construct $\phi_n:\mathcal{X}^n \rightarrow \left\{ 1,\dots, k_n\right\}$ such that $\phi^{-1}_n(i)=\mathcal{A}_{n,i}$ for all $1\leq i \leq k_n$. On the other hand, we can choose an arbitrary $y_{n,i}\in \mathcal{A}_{n.i}$ for each $i\in \Gamma_{k_n}$, and the mapping $\psi:\Gamma_{k_n} \rightarrow \mathcal{X}$ in the way $\psi(i)=y_{n,i}$.   
At this point, we observe:
\begin{IEEEeqnarray}{rCl}
	d(\phi_n, \psi_n,\mu^n)&=& \mathbb{P}(X\neq \psi_n(\phi_n(X))) \nonumber\\
	&=& \sum_{i=1}^{k_n} \sum_{x\in \mathcal{A}_{n,i}} f_\mu(x) \rho_{0,1}(x, y_{n, i}) \nonumber\\
	&=& \sum_{i=1}^{k_n} \mu(\mathcal{A}_{n,i} \setminus \left\{ y_{n,i} \right\}) \nonumber\\
	&=& \mu\left(\mathcal{X}\setminus \bigcup\limits_{i=1}^{k_n} \left\{ y_{n,i} \right\}\right).
	\label{eq_lemma_two_stage_1}
\end{IEEEeqnarray}
Then, from the hypothesis that $\left\{ \pi_n: n \geq 1\right\}$ is asymptotically sufficient for $\mu$, and the use of {\em bounded convergence theorem}, it is simple to verify that the RHS of (\ref{eq_lemma_two_stage_1}) goes to zero (see Section \ref{proof_lemma_zero_distortion}). Therefore, this convergence happens point-wise  $\forall \mu \in \Lambda$.

\begin{remark}
It worth mentioning that (\ref{eq_lemma_two_stage_1}) tends to zero,  if and only if,  $\lim_{n \rightarrow \infty} \cup_{k\geq n} \pi_k(x)= \left\{ x\right\}$ $\mu$-almost surely (cf. Lemma \ref{lemma_zero_distortion}). Hence, to achieve a point-wise convergence to zero of the distortion over $\Lambda$, for this two-stage scheme, it is necessary and sufficient that 
$\left\{ \pi_n: n \geq 1\right\}$ is asymptotically sufficient for $\Lambda$. 
\end{remark}

Regarding the second coding stage, we ideally need to  find a lossless code with the least worst-case average
redundancy over the family 
\begin{align*}
	\Lambda^n/\sigma(\pi_n) &\coloneqq  \left\{\mu^n/\sigma{(\pi_n \times \cdots \times \pi_n)}:\mu \in \Lambda \right\} \\ 
	&\subset \mathcal{P}(\mathcal{X}^n, \sigma(\pi_n \times \cdots \times \pi_n)),  
\end{align*}
where $\mu/\sigma{(\pi)} \in \mathcal{P}(\mathcal{X}, \sigma(\pi))$ is a short hand for the probability $\mu$ restricted to the sub-sigma field induced by $\pi$ \footnote{Note that if $\mu\in \mathcal{P}(\mathcal{X})$, then the restriction $\mu/\sigma(\pi) \in \mathcal{P}(\mathcal{X}, \sigma(\pi))$ reduces to the evaluation of $\mu$ over the cells of $\pi$. In fact, $\left\{\mu(B): B\in \pi \right\}$ plays the role of the probability mass function of $\mu/\sigma(\pi)$ on the measurable space $(\mathcal{X}, \sigma(\pi))$.},   and $P(\mathcal{X}, \sigma)$ denotes the collection of probabilities restricted to the events of the sub-sigma field $\sigma\subset \mathcal{B}(\mathcal{X})$. 

In fact, for a lossless prefix-free code $\mathcal{C}_n:\Gamma_{k_n}^n \rightarrow \left\{0,1\right\}^*$, associated to the first stage $\phi_n$, its worst-case average redundancy over $\Lambda^n$ 
is given by:
\begin{equation}\label{eq_lemma_two_stage_2}
	R(\Lambda^n, 
	\phi_n, \mathcal{C}_n) \coloneqq    \sup_{\mu \in \Lambda} \Big(r(\phi_n,\mathcal{C}_n, \mu^n ) - H(\mu)\Big).
\end{equation}
For any fixed $\mu\in \Lambda$, it is clear that the entropy of $\Phi_n(X^n)$ is a lower bound for the average rate of the code, i.e., $r(\phi_n,\mathcal{C}_n, \mu^n)\geq \frac{1}{n}H(\Phi_n(X^n))=H_{\sigma{(\pi_n)}}(\mu)$, then constraining to the events of $\sigma(\pi_n)$, we are interested in controlling the following stringer worst-case overhead:
\begin{equation}\label{eq_lemma_two_stage_3}
	\bar{R}(\Lambda^n,
	\phi_n,
	\mathcal{C}_n) \coloneqq    \sup_{\mu \in \Lambda} \big(r(\phi_n,\mathcal{C}_n, \mu^n ) - H_{\sigma(\pi_n)}(\mu)\big).
\end{equation}
Note that $H_{\sigma(\pi_n)}-H(\mu)\leq 0$ 
and thus, $\bar{R}(\Lambda^n, \phi_n,  \mathcal{C}_n)\geq R(\Lambda^n, \phi_n, \mathcal{C}_n)$. 
Then, we can choose a code solution to the following mini-max problem:
\begin{equation}\label{eq_lemma_two_stage_4}
	\mathcal{C}^*_n \coloneqq    \arg \min_{\mathcal{C}^n: \Gamma^n_{k_n} \rightarrow \left\{0,1\right\}^*} \bar{R}(\Lambda^n, 
	\phi_n, \mathcal{C}_n), \text{ for all } n\geq 1. 
\end{equation}

From the close connection between probabilities and prefix free codes \cite{csiszar_2004},   
the performance of the optimal code in (\ref{eq_lemma_two_stage_4}) is  tightly related to the i-radius of the family $\Lambda^n/\sigma(\pi_n)$  in (\ref{eq_sub_main_2}), in the sense that  $\forall n\geq 1$:
\begin{equation}\label{eq_lemma_two_stage_5}
	\frac{R^+(\Lambda^n, \sigma(\pi_n))}{n} \leq  \bar{R}(\Lambda^n,
	\phi_n, \mathcal{C}^*_n) \leq  \frac{R^+(\Lambda^n, \sigma(\pi_n))+2}{n}.
\end{equation}
Finally,  from the hypothesis on the information radius and (\ref{eq_lemma_two_stage_5}),  we have that:
\begin{IEEEeqnarray}{rCl}
\lim_{n \rightarrow \infty} R(\Lambda^n, \phi_n, \mathcal{C}^*_n) \leq \lim_{n \rightarrow \infty} \bar{R}(\Lambda^n, \phi_n, \mathcal{C}^*_n) =0.
\end{IEEEeqnarray}
\begin{remark}
The inequalities in (\ref{eq_lemma_two_stage_5}) states that the sub-linear trend (with the block-length) on the i-radius of $\left\{ \Lambda^n/\sigma(\pi_n): n\geq 1 \right\}$  is a necessary 
and sufficient condition for the existence of a strongly minimax universal  code for $\left\{ \Lambda^n/\sigma(\pi_n): n\geq 1 \right\}$.
\end{remark}
\end{proof}

\subsection{Theorem \ref{th_dichotomy_stronger_universality}}
\label{proof_th_dichotomy_stronger_universality}
Let us first introduce some notations, definitions and results that will be used in the proof of Theorem \ref{th_dichotomy_stronger_universality}. 

\begin{definition}\label{def_information_radius}
	For $\Lambda \subset \mathcal{P}(\mathcal{X})$ its i-radius is given and denoted by:
	$$R^+(\Lambda) \coloneqq   \inf_{v \in  \mathcal{P}(\mathcal{X})} \sup_{\mu \in \Lambda} \mathcal{D}(\mu | v).$$
\end{definition}

\begin{definition}\label{def_induced_prob_mapping}
	For a function $\phi: \mathcal{X} \rightarrow \mathcal{I}$  (where $\mathcal{I}$ is either a finite 
	or a countably infinite set) and $\mu\in \mathcal{P}(\mathcal{X})$, let us denote by $v_\mu\in \mathcal{P}(\mathcal{I})$
	the distribution induced by $\phi$ in $\mathcal{I}$ trough the standard construction\footnote{There is no question about the measurability
of $\phi(\cdot)$ as  we consider that $\mathcal{B(\mathcal{X})}$ is the power set.}: $v_\mu(B) \coloneqq \mu(\phi^{-1}(B))$ for all 
$B\subset \mathcal{I}$.
\end{definition}

\begin{lemma}\label{pro_invariance_envelop}
	Let  $\phi: \mathcal{X} \rightarrow \mathcal{I}$ be a mapping where $\mathcal{I}$ is a countably infinite set.
	Then for any non-negative envelope function $f: \mathcal{X} \rightarrow \mathbb{R}^+$, 
	there is $\tilde{f}:\mathcal{I} \rightarrow \mathbb{R}^+$ given by\footnote{$A_{i}\coloneqq  \phi^{-1}(\left\{ i \right\})$ for any $i\in \mathcal{I}$.}
	\begin{equation}\label{eq_th_DSU_5}
		\tilde{f}(i) \coloneqq    \min\ \left\{ \sum_{x\in \mathcal{A}_{i}} f(x), 1 \right\},
	\end{equation}
	such that $\left\{v_\mu: \mu\in \Lambda_f \right\} =\tilde{\Lambda}_{\tilde{f}} \coloneqq    \left\{ v\in \mathcal{P}(\mathcal{I}): f_v(i)\leq \tilde{f}(i)\text{ for all $i \in \mathcal{I}$} \right\}$. 
\end{lemma}

The proof of this result is presented in Appendix~\ref{proof_pro_invariance_envelop}.

Lemma \ref{pro_invariance_envelop} implies that envelope families on $\mathcal{X}$ map to envelope families on $\mathcal{I}$ through the mapping $\phi$. 
In this context, the result by Boucheron {\em et al.} \cite{boucheron_2009} in Theorem \ref{dichotomy_envelops} (in Section \ref{sec_background}) 
is instrumental to prove Theorem \ref{th_dichotomy_stronger_universality}.

{\em Proof of Theorem \ref{th_dichotomy_stronger_universality}:}
\subsubsection{Achievability}
If $f\in \ell_1(\mathcal{X})$, the fact that $\Lambda_f$ has a uniform bound on the tails of the distributions suggests
that a family of tail truncating partitions should be considered to achieve the claim i). Let us define
\begin{equation}\label{eq_th_DSU_1}
	\tilde{\pi}_{k_n}\coloneqq \left\{  \left\{ 1\right\}, \cdots,  \left\{ k_n\right\}, \Gamma_{k_n}^c \right\}, 
\end{equation}
which resolves  the elements of $\Gamma_{k_n}= \left\{1,\dots,k_n \right\}$ and, consequently,  
there is a pair $(\tilde{\phi}_n, \tilde{\psi}_n)$ associated with $\tilde{\pi}_{k_n}$ such that $\forall \mu \in \Lambda_f$:
\begin{equation}\label{eq_th_DSU_2}
	d(\tilde{\phi}_n, \tilde{\psi}_n,\mu^n) \leq \mu(\Gamma_{k_n}^c) \leq \sum_{x>k_n} f(x). 
\end{equation}
In fact, $\sup_{\mu \in \Lambda_f} d(\tilde{\phi}_n, \tilde{\psi}_n,\mu^n) \leq \sum_{x>k_n} f(x)<\infty$, and 
$(1/k_n)_n$ being  $o(1)$ is a sufficient condition to satisfy the uniform convergence of the distortion to zero.  Furthermore, 
from the proof of Lemma~\ref{lemma_two_stage_suffcond_al_usc} (Eq.(\ref{eq_lemma_two_stage_5})) and (\ref{eq_cor_universality_1}),  there is a lossless coding scheme $\left\{ \mathcal{\tilde{C}}_n:\Gamma_{k_n+1}^n  \rightarrow \left\{0,1\right\}^*: n\geq 1\right\}$ such that $\sup_{\mu\in \Lambda_f} \big(r(\phi_n, \mathcal{\tilde{C}}_n, \mu^n) -H(\mu)\big) \leq k_n\cdot \log \sqrt{n}/n + O(1/n)$. Therefore,  we can consider $(k_n)_n$ being $O(n^\tau)$ with $\tau \in (0,1)$ to conclude the achievability part.  

\subsubsection{Converse for Two-Stage Lossy Coding Schemes}\footnote{We first present this preliminary converse argument, as it provides the ground to explore the redundancy gain analysis presented in Section \ref{sec_redundancy_gains_envelop_families}.}
For the converse part,  let us first consider an arbitrary two-stage lossy rule  $(\phi_n, \psi_n, \mathcal{C}_n)$ with a finite partition $k_n =\left| \pi_n\right| < \infty$. If we denote its prototypes by $\mathcal{Y}_n \coloneqq    \left\{ \psi(i): i \in \Gamma_{k_n} \right\}$, it is clear that  there exists $\mu \in \Lambda_{f}$ such that $\textrm{supp}(\mu) \subset \mathcal{Y}_n^c$, and consequently,  $d(\phi_n, \psi_n,\mu^n)=1$ for all $n$.  Therefore, for any finite size partition rule $(\phi_n, \psi_n, \mathcal{C}_n)$ it follows that  $\sup_{\mu\in \Lambda_f} d(\phi_n, \psi_n,\mu^n)=1$ for all $n\geq 1$ and hence, when $f\notin \ell_1(\mathcal{X})$ no uniform convergence on the distortion can be achieved  with a finite size lossy rule.

On the other hand, for the family of infinite size partition rules, i.e., $(\phi_n, \psi_n, \mathcal{C}_n)$ such that $\left| \pi_n\right|= \infty$,  we focus our analysis on $R^+(\Lambda^n_f, \sigma(\pi_n))$ in (\ref{eq_sub_main_2}).  Let us fix a block-length $n>0$ and a rule $(\phi_n,\psi_n, \mathcal{C}_n)$
of infinite size. For sake of clarity, we consider that $\phi_n: \mathcal{X} \rightarrow \mathcal{I}$, where $\mathcal{I}$ is a $\infty$-alphabet.  
For any $\mu\in \Lambda_f$, $v_\mu$ denotes the induced measure in $\mathcal{I}$
by the mapping $\phi_n$ trough the standard construction (see Def. \ref{def_induced_prob_mapping}).  
In addition,  it is simple to verify 
that for any pair $\mu_1,\mu_2 \in \mathcal{P}(\mathcal{X})$ 
\begin{IEEEeqnarray}{rCl}
		\mathcal{D}_{\sigma(\pi_n)}(\mu_1 | \mu_2) &=& \mathcal{D}(v_{\mu_1}|v_{\mu_2})\nonumber\\
		&=&\sum_{i\in \mathcal{I}} f_{v_{\mu_1}}(i) \frac{f_{v_{\mu_1}}(i) }{f_{v_{\mu_2}}(i) }, \label{eq_th_DSU_3}
\end{IEEEeqnarray}
where $\pi_n=\left\{\mathcal{A}_{n,i}=\phi_n^{-1}(\left\{i \right\}): i\in \mathcal{I}\right\}$ and $f_{v_\mu}(i) \coloneqq    \mu(\mathcal{A}_{n,i})$
 $\forall i$ denotes the pmf of $v_\mu$ on $\mathcal{I}$. Then, 
\begin{IEEEeqnarray}{rCl}
	R^+(\Lambda^n_f, \sigma(\pi_n)) &=& R^{+}(\left\{v^n_\mu: \mu\in \Lambda_f \right\}) \label{eq_th_DSU_4}
\end{IEEEeqnarray}
where $v^n_\mu$ denotes de product probability on $\mathcal{I}^n$ with marginal $v_{\mu}$ and 
$\mathcal{P}(\mathcal{I}^n)$ is the collection of probability measures on $\mathcal{I}^n$. Then,  
the i-radius of $\Lambda_f^n$ restricted to the product sub-sigma field $\sigma(\pi_n\times \cdots \times \pi_n)$
is equivalent to the information radius of $\left\{v^n_\mu: \mu\in \Lambda_f \right\} \subset \mathcal{P}(\mathcal{I}^n)$ (Def. \ref{def_information_radius}).  
From Lemma \ref{pro_invariance_envelop},  $\left\{v^n_\mu: \mu\in \Lambda_f \right\}$ is an {envelope family} with envelope function  
given by (\ref{eq_th_DSU_5}).  It is simple to verify that $f\notin \ell_1(\mathcal{X})$ 
implies that $\tilde{f} \notin \ell_1(\mathcal{I})$,  then Theorem \ref{dichotomy_envelops} and (\ref{eq_th_DSU_4}) tell us
that  $R^+(\Lambda^n_f, \sigma(\pi_n))=\infty$. 
Finally, since the i-radius in (\ref{eq_th_DSU_4}) tightly bounds the least-worst expected redundancy for the second lossless coding stage (see  (\ref{eq_lemma_two_stage_4}) and (\ref{eq_lemma_two_stage_5})), this implies that: 
\begin{equation}\label{eq_th_DSU_6}
	\sup_{\mu\in \Lambda_f} (r(\phi_n, \mathcal{C}_n, \mu^n) -H(\mu)) = \infty, 
\end{equation}
which concludes the argument.  

\subsubsection{Converse for general variable-length lossy codes}
\label{converse_general_lossy_scheme}
Let us consider a general lossy code $(f_n,g_n)$ of length $n>0$ introduced in Section \ref{subsec:alsc}. Without loss of generality we can decouple  
$f_n$ as the composition of a vector quantizer $\phi_n:\mathcal{X}^n \rightarrow \mathcal{I}_n$, where 
$\mathcal{I}_n$ is an index set, and a prefix-free losses mapping $\mathcal{C}_n: \mathcal{I}_n  \rightarrow \left\{0,1\right\}^*$, 
where $f(x^n)=\mathcal{C}_n(\phi_n(x^n))$ for all $x^n\in \mathcal{X}^n$. From this, we characterize
the vector quantization induced by  $(f_n,g_n)$ as follows:
\begin{equation}\label{eq_th_DSU_7}
	\pi_n \coloneqq \left\{\phi_n^{-1}(\left\{ i \right\}): i \in \mathcal{I}_n \right\} \subset \mathcal{B}(\mathcal{X}^n).
\end{equation}
Using this two-stage (vector quantization-coding) view,  
it is possible to show that\footnote{The proof of (\ref{eq_th_DSU_7b}) is presented in Appendix~\ref{proof_pro_redundancy_bound}.}: 
\begin{align}\label{eq_th_DSU_7b}
\bar{R}(\Lambda_f,f_n) &\coloneqq    \sup_{\mu \in \Lambda_f} \Big(r(f_n,\mu^n) -\frac{1}{n}H_{\sigma(\pi_n)}(\mu^n)\Big) \nonumber\\
&\geq \frac{1}{n} \inf_{v\in \mathcal{P}(\mathcal{X}^n)} \sup_{\mu \in \Lambda_f} \mathcal{D}_{\sigma(\pi_n)}(\mu^n | v),  
\end{align}
which means that the worst-case overhead, expressed by $\bar{R}(\Lambda_f,f_n)$,  is lower bounded by the i-radius
 of the $n$-fold family $\Lambda_f^n$ projected into the sub-sigma field induced by $\pi_n$, i.e., a quantization of $\mathcal{X}^n$. 
Considering that $f\notin \ell_1(\mathcal{X})$, we follow the construction presented in \cite{boucheron_2009} that shows 
that there is an infinite collection of distributions $\tilde{\Lambda}=\left\{ \tilde{\mu}_j \in \Lambda_f, j\in \mathcal{J} \right\}$ with $\left| \mathcal{J}\right| =\infty$, where if we denote by 
$$\mathcal{A}_{\tilde{\mu}_j} \coloneqq    \textrm{supp}( \tilde{\mu}_j)=  \left\{ x \in \mathcal{X}: f_{\tilde{\mu}_j}(x)>0 \right\},$$ 
then $ \left| \mathcal{A}_{\tilde{\mu}_j} \right| < \infty$ for each $j\in \mathcal{J}$ and for any $j_1\neq j_2$, $\mathcal{A}_{\tilde{\mu}_{j_1}} \cap \mathcal{A}_{\tilde{\mu}_{j_2}} =\emptyset$. In this context, for each $j\in \mathcal{J}$ $A^n_{\tilde{\mu}_j} \coloneqq \mathcal{A}_{\tilde{\mu}_j} \times  \ldots \times \mathcal{A}_{\tilde{\mu}_j} \in \mathcal{X}^n$ is the support of $\tilde{\mu}^n_j$.   

At this point, let us use the assumption that:  
$\sup_{\mu \in \Lambda_f} d(f_n,g_n,\mu^n)<1.$
This implies that $\sup_{\mu \in \tilde{\Lambda}} d(f_n,g_n,\mu^n)<1$. From the fact that $\tilde{\Lambda}$ is an infinite collection of probabilities 
with disjoint supports and the definition of the distortion, it is simple to verify that we need to allocate at least one prototype\footnote{The prototypes of $(f_n,g_n)$ is the set $\mathcal{B}_n=\left\{g_n(f_n(x^n_1)): x^n_1\in \mathcal{X}^n \right\}$.} per cell $A^n_{\tilde{\mu}_j}$, which implies that $\left| \mathcal{I}_n \right| =\infty$,  because otherwise it follows that  $\sup_{\mu \in \tilde{\Lambda}} d(f_n,g_n,\mu^n)=1$.

Using (\ref{eq_th_DSU_7b}), we will focus on evaluating the information radius of $\tilde{\Lambda}^n$ projected over the measurable space $(\mathcal{X}^n,\sigma(\pi_n))$
considering that by definition:   
\begin{IEEEeqnarray}{rCl}
&\inf&_{v\in \mathcal{P}(\mathcal{X}^n)} \sup_{\mu \in \tilde{\Lambda}} \, \mathcal{D}_{\sigma(\pi_n)}(\mu^n | v) \nonumber\\
&\leq&  \inf_{v\in \mathcal{P}(\mathcal{X}^n)} \sup_{\mu \in {\Lambda}_f} \, \mathcal{D}_{\sigma(\pi_n)}(\mu^n | v). 
\label{eq_th_DSU_8b}
\end{IEEEeqnarray}
For every $j\in \mathcal{J}$, let us define the covering of the support of $\tilde{\mu}^n_{j}\in \tilde{\Lambda}^n$ by
$\pi_n(\mathcal{A}^n_{\tilde{\mu}_j}) \coloneqq    \left\{ B\in \pi_n: \mathcal{A}^n_{\tilde{\mu}_j} \cap B\neq \emptyset \right\}$
and 
\begin{align}\label{eq_th_DSU_9}
	\mathcal{B}(\mathcal{A}^n_{\tilde{\mu}_j}) \coloneqq    \bigcup_{B\in \pi_n(\mathcal{A}^n_{\tilde{\mu}_j})} B.
\end{align}
By construction, we note that $\left| \mathcal{A}^n_{\tilde{\mu}_j} \right| <\infty$ and consequently $\left| \pi_n(\mathcal{A}^n_{\tilde{\mu}_j}) \right| <\infty$
for all $j\in \mathcal{J}$. Considering that $\pi_n$ has an infinite number of cells, we can choose an infinite subset 
of elements in $\tilde{\Lambda}^n=\left\{\tilde{\mu}^n_j: j\in \mathcal{J} \right\}$ in the following way:  We fix $j_1=1$ and $\bar{\mu}_1=\tilde{\mu}_{1}\in \tilde{\Lambda}$,  then we consider
\begin{equation}\label{eq_th_DSU_10}
	j_2 =  \min \left\{ j>j_1, \text{ such that } \mathcal{B}(\mathcal{A}^n_{\tilde{\mu}_j}) \cap \mathcal{B}(\mathcal{A}^n_{\bar{\mu}_1}) =\emptyset \right\} <\infty,
\end{equation}
and we choose $\bar{\mu}_2 = \tilde{\mu}_{j_2}\in \tilde{\Lambda}$. Iterating this rule,   
at the $k$-stage ($k\geq 2$) we solve
\begin{equation}\label{eq_th_DSU_11}
	j_k =  \min \left\{ j>j_{k-1}, \text{ such that } \mathcal{B}(\mathcal{A}^n_{\tilde{\mu}_j}) \cap \left(\bigcup\limits_{l=1}^{k-1} \mathcal{B}(\mathcal{A}^n_{\bar{\mu}_l})\right) =\emptyset \right\} 
\end{equation}
and we take $\bar{\mu}_k = \tilde{\mu}_{j_k}\in \tilde{\Lambda}$, for all $k\geq 1$.  Note that the solution of (\ref{eq_th_DSU_11}) is guaranteed from the fact that $ \left| \mathcal{J}\right| =\infty$ and $\left| \pi_n(\mathcal{A}^n_{\tilde{\mu}_j}) \right| <\infty$ for all $j$. Finally, we define  $\bar{\Lambda}^n \coloneqq    \left\{  \bar{\mu}^n_l: l\geq 1\right\} \subset \tilde{\Lambda}^n$. Importantly (for the computation of the i-radius), this restricted family of distributions has the property that their support coverings in (\ref{eq_th_DSU_9}) are disjoint by its construction in (\ref{eq_th_DSU_11}). From $\bar{\Lambda}^n$, we can induce the following partition:
\begin{equation}\label{eq_th_DSU_12}
	\xi_n \coloneqq     \left\{ \mathcal{B}(A^n_{\bar{\mu}_l}): l\geq 1) \right\} \cup \left(\mathcal{X}^n\setminus \bigcup_{l\geq 1} \mathcal{B}(A^n_{\bar{\mu}_l})\right)  \subset \sigma(\pi_n),
\end{equation}
where the last identity is from the construction,  as  every cell of $\xi_n$ is a finite union 
of cells of $\pi_n$ (i.e., $\xi_n \ll \pi_n$). It is not difficult to check that for every $v\in \mathcal{P}(\mathcal{X}^n)$,
we have that~\footnote{This result follows from the fact that the elements of $\bar{\Lambda}^n$ projected into the 
sub-sigma field $\xi_n$ degenerate, in the sense that  $H_{\sigma(\xi_n)}(\bar{\mu}^n_l)=0$ for all $l\geq 1$.}
\begin{equation}\label{eq_th_DSU_12b}
	\sup_{\bar{\mu}_l\in \bar{\Lambda}}\mathcal{D}_{\sigma(\xi_n)}	(\bar{\mu}^n_{l}|v)=\infty.
\end{equation}
Consequently, we have that 
\begin{align}\label{eq_th_DSU_12c}
	\inf_{v\in  \mathcal{P}(\mathcal{X}^n)} \sup_{\tilde{\mu}_j \in \tilde{\Lambda}} \mathcal{D}_{\sigma(\pi_n)}(\tilde{\mu}_j^n | v) &\geq  \inf_{v\in  \mathcal{P}(\mathcal{X}^n)} \sup_{\tilde{\mu}_j \in \tilde{\Lambda}} \mathcal{D}_{\sigma(\xi_n)}(\tilde{\mu}_j^n | v)\nonumber\\
	&\geq\inf_{v\in  \mathcal{P}(\mathcal{X}^n)} \sup_{\bar{\mu}_l\in \bar{\Lambda}}\mathcal{D}_{\sigma(\xi_n)}	(\bar{\mu}^n_{l}|v)=\infty,
\end{align}
the first inequality derives from $\xi_n \ll \pi_n$ and the second from  $\bar{\Lambda}\subset \tilde{\Lambda}$. 
Finally (\ref{eq_th_DSU_12c}) and the relationship between  the worst-case redundancy and the information radius in (\ref{eq_th_DSU_7b}) (Prop.~\ref{pro_redundancy_bound} in Appendix \ref{proof_pro_redundancy_bound}) imply that
	\begin{align}\label{eq_th_DSU_13}
		\bar{R}(\tilde{\Lambda},f_n)=  \sup_{\tilde{\mu}_j \in \tilde{\Lambda} } \Big(r(f_n,\tilde{\mu}_j^n) -\frac{1}{n}H_{\sigma(\pi_n)}(\tilde{\mu}_j^n)\Big)=\infty. 
	\end{align}
	In other words, from (\ref{eq_th_DSU_13}) there is $jo\in \mathcal{J}$ such that $r(f_n,\tilde{\mu}_{jo}^n) -\frac{1}{n}H_{\sigma(\pi_n)}(\tilde{\mu}_{jo}^n)=\infty$, where considering that by construction $\frac{1}{n}H_{\sigma(\pi_n)}(\tilde{\mu}_{jo}^n) \leq H(\tilde{\mu}_{jo})< \log \left| \mathcal{A}_{\tilde{\mu}_{jo}} \right| <\infty$, this implies that $r(f_n,\tilde{\mu}_{jo}^n) - H(\tilde{\mu}_{jo})=\infty$. Therefore, we have that 
	\begin{align}\label{eq_th_DSU_14}
		\sup_{\tilde{\mu}_j \in \tilde{\Lambda} } \Big(r(f_n,\tilde{\mu}_j^n) -H(\tilde{\mu}_j^n)\Big)=\infty,  
	\end{align}
	which concludes the result considering that $\tilde{\Lambda}\subset \Lambda_f$.
\hspace{\fill}~\QED

\subsection{Theorem \ref{th_information_radius_gain_tail_schemme}}
\label{proof_th_information_radius_gain_tail_schemme}
Without loss of generality,  in this section we assume that $\mathcal{X}$ is the integer set $\mathbb{N}\setminus \left\{ 0\right\}$.  
To organize the proof, we first introduce some definitions and a series of important results that will be used in the main argument.
\subsubsection{Preliminaries}
\label{proof_th_information_radius_gain_tail_schemme_preli}
\begin{definition}\label{def_hazard_function} \cite{bontemps_2014}
	For a non-negative function $f: \mathcal{X} \rightarrow \mathbb{R}^+$, the {\em hazard function} of $\Lambda_f$ is given by 
	$$h_f(u) \coloneqq    - \ln \bar{F}_f(u)$$ 
	for all $u \in \mathcal{X}$.
\end{definition}
\begin{definition}\label{def_hazard_function_continuos}\cite{bontemps_2014}
The continuous extension of $(h_f(u))_{u \in \mathcal{X}}$ 
to the positive real line $\mathbb{R}^+$ is defined by means of the following linear interpolation\footnote{This idea was proposed by Bontemps {\em et al.} \cite{bontemps_2014} following Anderson \cite{anderson_1970}.}:  
$$\tilde{h}_f(k\lambda + (1-\lambda)(k+1)) \coloneqq    \lambda h_f(k) + (1-\lambda) h_f(k+1)$$ 
for  $k\in \mathcal{X}$ and for all $\lambda \in [0,1)$.  
\end{definition}

Consistently with $(\tilde{h}_f(x))_{x\geq 0}$,  in Def.\ref{def_hazard_function_continuos}, it is possible to extend $(F_f(u))_{u\in \mathcal{X}}$ to $\mathbb{R}^+$ using the 
relationship expressed in Def. \ref{def_hazard_function}: 
\begin{definition}\label{def_continuos_envelop}
Given $f: \mathcal{X} \rightarrow \mathbb{R}^+$, the continuous extension of $(F_f(u))_{u\in \mathcal{X}}$ using $(\tilde{h}_f(x))_{x\geq 0}$ (Def. \ref{def_hazard_function_continuos}) is denoted by $(\mathcal{F}_f(x))_{x\geq 0}$ and called the {\em smoothed envelope distributions} of $\Lambda_f$. 
\end{definition}
\begin{definition}\label{def_continuos_Uf}\cite[Eq.(1)]{bontemps_2014}
Under the setting of Def. \ref{def_continuos_envelop}, a function $U_f: [1,\infty]  \longrightarrow \mathbb{R}$ can be defined 
as the solution of:
\begin{equation}\label{eq_th_IRG_tail_schemme_1}
	U_f(t) \coloneqq    \mathcal{F}^{-1}_f \left(1 - {1}/{t} \right),
\end{equation}
for all $t\geq 1$. 
\end{definition}
\begin{definition}\label{def_function_lf}
	Let $(f(x))_{x\in \mathcal{X}}$ be non-negative and in $\ell_1(\mathcal{X})$.  A non-decreasing continuous function can 
	be obtained as\footnote{See Eq.(\ref{eq_th_IRG_tail_schemme_11}).}: 
	$$l_f(1/\epsilon) \coloneqq \int^{1/\epsilon^2}_{1} \frac{U_f(x)}{2x} \partial x,$$ 
	for any $\epsilon>0$. 
\end{definition}
\begin{definition}\label{def_epsilon_f}
	Let  $(\epsilon_{f,n})_{n\geq 1}$ be the sequence obtained as the solution (point-wise) of:  $l_f(1/\epsilon)=\frac{n\epsilon^2}{8}$ for all $n$.
\end{definition}

We are in the position to state two instrumental results: 
\begin{lemma} \cite[Th. 4]{boucheron_2009} \label{th_upper_bound_radius}
	Let $\left\{ \Lambda^n_f, n\geq 1 \right\} $ be the envelope collection of stationary and memoryless sources with $f\in \ell_1(\mathcal{X})$ and tail function $(\bar{F}_f(u))_{u\in \mathcal{X}}$. Then for any $n\geq 1$
	\begin{equation*}
		R^+(\Lambda_f^n) \leq \inf_{u\geq 1} \left[ n \bar{F}_f(u) \log (e) + \frac{u-1}{2} \cdot \log n \right] +2.
	\end{equation*}
\end{lemma}

\begin{lemma} \cite[Prop. 5]{bontemps_2014} \label{th_lower_bound_radius}
	Under the setting of Lemma \ref{th_upper_bound_radius}, there is a sequence $(\xi_n)_n$ being $o(1)$ (and function of $f$) 
	such that:\footnote{Remarkably, it has been shown in \cite[Th. 2]{bontemps_2014} that this closed-form lower bound captures the precise asymptotic of the information radius of the envelope class, meaning that: $\lim_{n \rightarrow \infty} R^+(\Lambda_f^n)/\log (e) \int_1^n \frac{U_f(x)}{2x} \partial x=1$.}
	\begin{equation}
		R^+(\Lambda_f^n) \geq (1+ \xi_n) \log (e) \int_1^n \frac{U_f(x)}{2x} \partial x, \text{ for all $n\geq 1$.}
	\end{equation}
\end{lemma}

We also use results from the seminal work of Haussler and Opper \cite{haussler_1997} that we summarize here: 
\begin{definition}\label{def_hellinder_distance}
	For any $\mu_1$, $\mu_2 \in \mathcal{P}(\mathcal{X})$, the  {\em Hellinger distance}  is given/denoted by: $d_h(\mu_1.\mu_2)^2 \coloneqq    \sum_{x\in \mathcal{X}} (\sqrt{f_{\mu_1}(x)}- \sqrt{f_{\mu_2}(x)})^2$.
\end{definition}

\begin{definition} \label{metric_entropy_def} \cite{haussler_1997}
	For $\Lambda\subset \mathcal{P}(\mathcal{X})$ and $\epsilon>0$,  let $\mathcal{D}_\epsilon(\Lambda)$ be the smallest  cardinality of a partition of $\Lambda$, whose cells have a diameter smaller or equal then $\epsilon$ (with respect to $d_h$ in $ \mathcal{P}(\mathcal{X})$) or it  is infinity if no finite partition satisfies the diameter condition. Then, the metric entropy of $\Lambda$ is given by: 
	$$\mathcal{H}_\epsilon(\Lambda) \coloneqq \ln (\mathcal{D}_\epsilon(\Lambda)).$$ 
\end{definition}
The following important results can be stated:
\begin{lemma}\cite[Lemma 7]{haussler_1997}
	\label{lemma_HO_metric_entropy}
	Let us assume that $\Lambda \subset \mathcal{P}(\mathcal{X})$ is totally bounded, i.e., $\mathcal{H}_\epsilon(\Lambda)< \infty $ for all $\epsilon>0$.  
	Then, for all $n\geq 1$, 
	\begin{equation*}
		R^+({\Lambda}^n) \geq \log (e) \cdot \sup_{\epsilon>0} \min \left\{ \mathcal{H}_\epsilon(\Lambda),\frac{n\epsilon^2}{8}  \right\} -1. 
	\end{equation*}
\end{lemma}

\begin{corollary}\label{cor_lemma_HO_metric_entropy}
From Lemma \ref{lemma_HO_metric_entropy},  if we let $\epsilon^*_{\Lambda,n} \coloneqq    \inf \left\{ \epsilon>0: \mathcal{H}_\epsilon(\Lambda) \leq \frac{n\epsilon^2}{8} \right\}$, we have that $\forall n\geq 1$:
\begin{equation}\label{eq_th_IRG_tail_schemme_4}
	R^+({\Lambda^n}) \geq  \log (e) \cdot \mathcal{H}_{\epsilon^*_{\Lambda,n}}(\Lambda)-1, 
\end{equation}
and,  consequently, 
$$\lim \inf_{n \rightarrow \infty} R^+({\Lambda^n})/\log (e) \mathcal{H}_{\epsilon^*_{\Lambda,n}}(\Lambda) \geq 1.$$ 
\end{corollary}
It is worth noting that the metric entropy lower bound for the information radius stated in Corollary \ref{cor_lemma_HO_metric_entropy} is asymptotically tight under a {\em slowly variant condition} on the behaviour of $\mathcal{H}_{\epsilon}(\Lambda)$ as $\epsilon$ goes to zero \cite{bontemps_2014}.\footnote{More details are presented in  \cite[Lem. 8, Th. 4 and Th. 5]{haussler_1997}.}

Importantly for envelope families,  when $f\in \ell_1(\mathcal{X})$ the asymptotic of the metric entropy of  $\Lambda_f$,  
i.e., $\lim_{\epsilon \rightarrow 0} \mathcal{H}_{\epsilon}(\Lambda_f)$,  is known. More precisely, Bontemps {\em et al.}\cite[Prop.4]{bontemps_2014} have shown that
\begin{equation}\label{eq_th_IRG_tail_schemme_11}
	\mathcal{H}_{\epsilon}(\Lambda_f)= (1+o_f(1)) \int_1^{1/\epsilon^2} \frac{U_f(x)}{2x} \partial x
\end{equation}
as $\epsilon$ tends to $0$. 

\subsubsection{Proof of Theorem \ref{th_information_radius_gain_tail_schemme}--- Regime of Gain in Minimax Redundancy}
\label{proof_th_information_radius_gain_tail_schemme_gain_regime}
%
Let us assume
that $(k_n)_n$ is $o(u^*_f(n))$. This part derives directly from the tight lower and upper bounds  developed by Bontemps {\em et al.} \cite[Th. 2]{bontemps_2014} and Boucheron {\em et al.} \cite[Th. 4]{boucheron_2009} for the case of summable envelopes.  In particular, from Lemmas~\ref{th_upper_bound_radius} we have 
that 
\begin{align*}
R^+(\Lambda_f^n) &\leq \left[ n \bar{F}_f(u^*_f(n)) \log (e) + \frac{u^*_f(n)-1}{2} \log n \right] +2 \\
	&\leq  2+ \log (e) + \frac{u^*_f(n)-1}{2} \log n. 
\end{align*}
On the other hand,  it has been shown that \footnote{Notice that: $\int_1^n  \frac{U_f(x)}{2x } \partial x=\frac{1}{2} \int_0^{\ln n} U_f(e^y) \partial y \geq \frac{U_f(n)\ln n}{4}$, the last inequality from the concavity and positivity of $U_f(e^y)$ shown in \cite[pp. 814]{bontemps_2014}. On the other hand,  from their definitions $\forall n\in \mathcal{X}$,  $u^*_f(n)-1 \leq U_f(n) < u^*_f(n)$.}
	\begin{equation}
\int_1^n \frac{U_f(x)}{2x} \partial x \geq \frac{U_f(n)\log n}{4} \geq \frac{ (u^*_f(n)-1)}{4} \log n.
	\end{equation}
Consequently, from Lemma \ref{th_lower_bound_radius} we have that eventually with $n$:
\begin{align}\label{eq_th_IRG_tail_schemme_2}
	(1+ \xi_n)&\frac{ (u^*_f(n)-1) }{4} \log n	\leq R^+(\Lambda_f^n) \nonumber\\
	&\leq 2+ \log (e) + \frac{(u^*_f(n)-1)}{2} \log n, 
\end{align}
which means that $(R^+(\Lambda_f^n))_{n} \approx (u^*_f(n) \log n)$. 
Moreover, it is well-known that~\cite{csiszar_2004}:
\begin{equation}\label{eq_th_IRG_tail_schemme_3}
R^+(\Lambda_f^n, \sigma(\tilde{\pi}_{k_n})) \leq R^+(\mathcal{P}^n(k_n)) \leq \frac{k_n-1}{2} \log n + K, 
\end{equation}
for some $K>0$, where  $\mathcal{P}(k_n)$ is a short-hand for the collection all probabilities defined on the finite alphabet $\Gamma_{k_n}$, i.e., the simplex of dimension $k_n-1$. Consequently, under the assumption that $(k_n)_n$ is $o(u^*_f(n))$,  from (\ref{eq_th_IRG_tail_schemme_2}) and (\ref{eq_th_IRG_tail_schemme_3}) it follows that: 
\begin{equation*}
	\lim_{n \rightarrow \infty} \frac{R^+(\Lambda_f^n, \sigma(\tilde{\pi}_{k_n}))}{R^+(\Lambda_f^n)} =0.
\end{equation*}
\hspace{\fill}~\QED

\subsubsection{Proof of Theorem \ref{th_information_radius_gain_tail_schemme} --- Regime of No-gain in Minimax Redundancy}
\label{proof_subsec_no_gain_i_radius}
Let us assume that $(k_n)_n \geq (u^*_f(n))_n$ eventually with $n$.  Here we  adopt results from the seminal work of Haussler and Opper \cite{haussler_1997} that offers a lower bound for the mutual information and consequently,  the channel capacity that corresponds to the information radius of a family of distributions \cite{csiszar_2004}. 
However in our problem, 
we have a dynamic collection of distributions, explained 
by the process of projecting $\Lambda_f$ into the dynamic collection of sub-sigma fields $\left\{ \sigma(\tilde{\pi}_{k_n}): n\geq 1\right\}$.  More precisely, and adopting the notation introduced  in Section \ref{proof_lemma_two_stage_suffcond_al_usc}, we have the collection of distributions:  
\begin{align}\label{eq_th_IRG_tail_schemme_5}
	\Lambda_f^n/\sigma(\tilde{\pi}_{k_n}) &\coloneqq    \big\{\mu^n/\sigma{(\tilde{\pi}_{k_n} \times \cdots \times \tilde{\pi}_{k_n})}:\mu \in \Lambda_f \big\} \nonumber\\
	&\subset \mathcal{P}(\mathcal{X}^n, \sigma(\tilde{\pi}_{k_n} \times \cdots \times \tilde{\pi}_{k_n})), 
\end{align}
for all  $n\geq 1$, where $\mu/\sigma{(\pi)} \coloneqq \left\{\mu(\mathcal{A}): \mathcal{A}\in \sigma{(\pi)} \right\} \in \mathcal{P}(\mathcal{X}, \sigma(\pi))$ denotes the probability $\mu$  restricted to the sub-sigma field induced by $\pi$ and  $\mathcal{P}(\mathcal{X}, \sigma)$ denotes the collection of probabilities restricted to the events of the sub-sigma field $\sigma\subset \mathcal{B}(\mathcal{X})$. Furthermore, associated to $\tilde{\pi}_{k_n}= \left\{ \mathcal{A}_{k_n,i}: i=1,..,k_n\right\}$ there is a lossy mapping  $\phi_n: \mathcal{X}^n \longrightarrow \Gamma_{k_n}$ where $\phi_n^{-1}(i)=\mathcal{A}_{k_n,i}$ for $i\in \Gamma_{k_n}$.  
Consequently through $\phi_n$,  every $\mu \in \mathcal{P}(\mathcal{X})$ induces a probability in $\Gamma_{k_n}$, which we denote by  $\rho_{\mu,k_n} \in \mathcal{P}(\Gamma_{k_n})$, by the standard construction: $\rho_{\mu,k_n}(B)=\mu(\phi_n^{-1}(B))$ for all $B\subset \Gamma_{k_n}$.  Note that $\rho_{\mu,k_n}$ is fully characterized by its pmf $f_{\rho_{\mu,k_n}}(i)=\rho_{\mu,k_n}( \left\{ i \right\})=\mu(\mathcal{A}_{k_n, i})$, $\forall i \in \Gamma_{k_n}$, where we have that  $f_{\rho_{\mu,k_n}}(i)=f_\mu(i)$ if $i<k_n$ and $f_{\rho_{\mu,k_n}}(k_n)=1-\mu(\Gamma_{k_n-1})$.
By letting 
$$\tilde{\Lambda}_{f,k_n} \coloneqq    \left\{ \rho_{\mu,k_n}:\mu \in \Lambda_f \right\} \subset \mathcal{P}(\Gamma_{k_n}),$$  
from (\ref{eq_th_DSU_4}) we have that:
\begin{equation}\label{eq_th_IRG_tail_schemme_6}
	R^+(\Lambda^n_f, \sigma(\tilde{\pi}_{k_n})) =  R^+(\tilde{\Lambda}^n_{f,k_n}).
\end{equation}
Consequently, the problem reduces to characterize the information radius of a family of dynamic distributions 
$\left\{\tilde{\Lambda}^n_{f,k_n}, n\geq 1 \right\}$ (defined on a dynamic alphabet whose size grows with the block-length). Using the envelope conditions of $\Lambda_f$ and Lemma \ref{pro_invariance_envelop} in Section \ref{proof_th_dichotomy_stronger_universality}, it is simple to show that $\tilde{\Lambda}_{f,k_n}$ satisfies an envelope condition on $\mathcal{P}(\Gamma_{k_n})$, which is expressed in (\ref{eq_th_IRG_tail_schemme_7}). 
\begin{figure*}[!t]
\begin{equation}\label{eq_th_IRG_tail_schemme_7}
	\tilde{\Lambda}_{f,k_n}= \left\{\rho \in  \mathcal{P}(\Gamma_{k_n}): 
	f_{\rho}(i)
	\leq f(i), \text{ for } i=[1:k_{n-1}] \text{ and } 
	f_{\rho}(k_n) 
	\leq \sum_{l > k_{n} -1}f(l)=\bar{F}_f(k_n-1) \right\}.
\end{equation}
\hrulefill 
\end{figure*}
Then, if we consider the extended (over the integer) finite size envelope function $\tilde{f}_{k_n}:\mathcal{X} \longrightarrow \mathbb{R}^+$ given by: $\tilde{f}_{k_n}(i)\coloneqq f(i)$ for $ i=[1:k_{n}-1]$, $\tilde{f}_{k_n}(k_n) \coloneqq \bar{F}_f(k_n-1)$ and $\tilde{f}_{k_n}(i) \coloneqq 0$ for $i>k_n$, 
$\tilde{\Lambda}_{f,k_n}$ is equivalent to ${\Lambda}_{\tilde{f}_{k_n}} \subset \mathcal{P}(\mathcal{X})$ and thus, 
\begin{equation}\label{eq_th_IRG_tail_schemme_8}
	R^+(\Lambda^n_f, \sigma(\tilde{\pi}_{k_n})) =  R^+(\tilde{\Lambda}^n_{f,k_n}) = R^+( \Lambda^n_{\tilde{f}_{k_n}}), \ \forall n\geq 1.
\end{equation}
Therefore, studying the minimax redundancy gain reduces to analyze the family of envelope distributions of finite size $\left\{ \Lambda_{\tilde{f}_{k_n}}: n\geq 1 \right\}$, where $\textrm{supp}(\tilde{f}_{k_n}) \subset \Gamma_{k_n}$ by construction. 
If we consider, 
\begin{equation}\label{eq_th_IRG_tail_schemme_9}
	\epsilon^*_{n,k} \coloneqq  \inf \left\{ \epsilon>0: \mathcal{H}_\epsilon(\Lambda_{\tilde{f}_{k_n}}) \leq \frac{n\epsilon^2}{8} \right\}, 
\end{equation}
the straight adoption of Lemma \ref{lemma_HO_metric_entropy} in this dynamic context implies that 
$$R^+( \Lambda^n_{\tilde{f}_{k_n}}) \geq  \log (e) \cdot \mathcal{H}_{\epsilon^*_{n,k_n}}(\Lambda_{\tilde{f}_{k_n}})-1$$ 
for all $n$ and, consequently, 
\begin{equation}\label{eq_th_IRG_tail_schemme_10}
\liminf_{n \longrightarrow \infty} R^+( \Lambda^n_{\tilde{f}_{k_n}}) \geq  \log (e) \cdot \liminf_{n \longrightarrow \infty} \mathcal{H}_{\epsilon^*_{n,k_n}}(\Lambda_{\tilde{f}_{k_n}})-1.
\end{equation} 
Following the approach proposed by Haussler {\em et al.} \cite{haussler_1997}, the idea is to obtain a tight approximation (ideally in closed-form) of the RHS of (\ref{eq_th_IRG_tail_schemme_10}), assuming that the function $\mathcal{H}_{\epsilon^*_{n,k_n}}(\Lambda_{\tilde{f}_{k_n}})$ is 
asymptotically  lower bounded by a continuous non-decreasing function. 
With that objective in mind, the following important result (Theorem \ref{pro:no_gain_minimax_redundancy} below) can be obtained.
For the statement of this result, the following definition is needed:
\begin{definition}\label{admisibility_envelop_given_size}
	Given $(f(x))_{x\in \mathcal{X}}$, non-negative and in $\ell_1(\mathcal{X})$,  and a sequence of positive integers $(k_n)_n$, we say that $(\epsilon_n)_n \in (\mathbb{R}^+\setminus \left\{ 0\right\})^{\mathbb{N}}$ is admissible for $(k_n)_n$ given $\Lambda_f$ if
	\begin{equation} \label{eq_admisibility_condition_given_the_size}
		\bar{F}_f(k_n-1) \leq \frac{\epsilon^2_{n}}{16} 
	\end{equation}
	holds eventually (with $n$). 
\end{definition}

\begin{theorem}\label{pro:no_gain_minimax_redundancy}
Let us consider 
$\Lambda_f \subset \mathcal{P}(\mathcal{X})$,  with $f \in \ell_1(\mathcal{X})$ and $\textrm{supp}(f)=\mathcal{X}$,  and a sequence of non-decreasing positive integers $(k_n)_n$ such that $(1/k_n)_n$ is $o(1)$.
If $(\epsilon_{f,n})_n$ (see Def. \ref{def_epsilon_f}) is admissible for $(k_n)_n$ given $\Lambda_f$ (see Def. \ref{admisibility_envelop_given_size}) 
then 
$$\liminf_{n \rightarrow \infty} \frac{R^+( \Lambda^n_{\tilde{f}_{k_n}})}{\log (e) \cdot \int^{n}_{1} \frac{U_f(x)}{2x} \partial x} \geq 1$$
\end{theorem}

The proof of Theorem \ref{pro:no_gain_minimax_redundancy} is presented in Section \ref{proof_pro:no_gain_minimax_redundancy}.

\begin{remark} 
\label{remark_about_no_gain_minimax_redundancy}
In general we have that  
$$R^+( \Lambda^n_{\tilde{f}_{k_n}}) = R^+( \Lambda^n_f, \sigma(\tilde{\pi}_{k_n})) \leq R^+(\Lambda^n_f),$$
 the last inequality from (\ref{eq_th_IRG_tail_schemme_8}). On the other hand, it is known from \cite[Th. 2]{bontemps_2014}  
 that there is a sequence $(a_n)_n$ being $o(1)$  where eventually in $n$
 \begin{align} \label{eq_th_IRG_tail_schemme_14_pre}
 R^+{(\Lambda^n_f)} &\leq (1+a_n) \log (e) \cdot l_f(\sqrt{n}) \nonumber\\
 &=(1+a_n) \log (e) \cdot \int^{n}_{1} \frac{U_f(x)}{2x} \partial x.
 \end{align}
Consequently, under the assumptions of  Theorem \ref{pro:no_gain_minimax_redundancy} it follows directly from this result 
and (\ref{eq_th_IRG_tail_schemme_14_pre}) that
\begin{equation} \label{eq_th_IRG_tail_schemme_14}
	\lim_{n \rightarrow \infty} \frac{R^+(\Lambda^n_{\tilde{f}_{k_n}})}{R^+(\Lambda^n_f)}=1.
\end{equation}
\end{remark}

Returning to the proof, from  Theorem \ref{pro:no_gain_minimax_redundancy}, Definition \ref{admisibility_envelop_given_size} and Remark \ref{remark_about_no_gain_minimax_redundancy},  a sufficient condition to obtain (\ref{eq_th_IRG_tail_schemme_14}) (i.e., no gain in minimax redundancy) is that $(\sqrt{\bar{F}_f(k_n)})_n \ll (\epsilon_{f,n})_n$ (see some remarks about this in Lemma \ref{proposition_metric_entropy_asymtotics}, Section \ref{proof_pro:no_gain_minimax_redundancy} below). Furthermore from the proof of Theorem \ref{pro:no_gain_minimax_redundancy}, we have that $(\epsilon_{f,n}) = (\sqrt{8/n \cdot l(1/\epsilon_{f,n})})_n \sim (\sqrt{8/n \cdot l_f(\sqrt{n})})_n = (\sqrt{8/n \cdot \int_1^n \frac{U_f(x)}{2x} \partial x})_n$, where it is known that $\int_1^n \frac{U_f(x)}{2x} \partial x\geq \frac{U_f(n) \ln n}{4}$ \cite[pp.814]{bontemps_2014}. 
From the main assumption,  which consider that there is $N_o>0$ such that $\forall n  \geq N_o$, $k_n \geq u^*_f(n)$, we note that
$$ 
\left( \frac{\bar{F}_f(k_n)}{(\epsilon_{f,n})^2} \right)_n  \sim  \left( \frac{n \cdot \bar{F}_f(k_n)}{\int_1^n \frac{U_f(x)}{2x} \partial x} \right)_n,
$$ 
where for the second series we have that:
\begin{IEEEeqnarray}{rCl} 
	\frac{n \cdot \bar{F}_f(k_n)}{\int_1^n \frac{U_f(x)}{2x} \partial x} &\leq & \frac{4 n \cdot \bar{F}_f(k_n)}{U_f(n) \ln n} \nonumber\\
	&\leq & \frac{4 n \cdot \bar{F}_f(u^*_f(n))}{U_f(n) \ln n} \nonumber\\
	&< & \frac{4}{U_f(n) \ln n} \longrightarrow 0.  \label{eq_th_IRG_tail_schemme_15}
\end{IEEEeqnarray}
The strict inequality in (\ref{eq_th_IRG_tail_schemme_15}) is by definition of $u^*_f(n)$ in (\ref{eq_sec_rgef_4}), where $1/n \in (\bar{F}_f(u^*_f(n)), \bar{F}_f(u^*_f(n)-1)]$.  The last convergence in the RHS of (\ref{eq_th_IRG_tail_schemme_15}) is from the fact that  $U_f(n)\in [u^*_f(n)-1, u^*_f(n)) \rightarrow \infty$ as $n$ tends to infinity,  this follows from (\ref{eq_sec_rgef_4}) and the non-trivial assumption that $\left| \textrm{supp}(f) \right| =\infty$.  
In summary, from (\ref{eq_th_IRG_tail_schemme_15}) we have that  $\left(\sqrt{\bar{F}_f(k_n)}\right)_n \ll (\epsilon_{f,n})_n$, then Theorem \ref{pro:no_gain_minimax_redundancy} and its corollary in  (\ref{eq_th_IRG_tail_schemme_14}) implies that $\lim_{n \rightarrow \infty} \frac{R^+(\Lambda^n_{\tilde{f}_{k_n}})}{R^+(\Lambda^n_f)}=1$. 
This last limit and the equalities in (\ref{eq_th_IRG_tail_schemme_8}) conclude the proof. 
\hspace{\fill}~\QED

\subsection{Theorem \ref{pro:no_gain_minimax_redundancy}}
\label{proof_pro:no_gain_minimax_redundancy}
To organize the proof of Theorem \ref{pro:no_gain_minimax_redundancy},  we present first two instrumental results: 

The first result is a simple extension of \cite[Lemma 8]{haussler_1997}:
\begin{lemma}\label{lemma_LB_redundancy_metric_entropy_dynamic}
	Let us consider the  dynamic collection of distributions $\left\{ \Lambda_{\tilde{f}_{k_n}}: n\geq 1 \right\}$ presented in (\ref{eq_th_IRG_tail_schemme_7}) where $f\in \ell_1(\mathcal{X})$, and let $(k_n)_n$  be a non-decreasing sequence of integers. In addition, let $l: \mathbb{R}^+ \rightarrow \mathbb{R}^+$ be a strictly increasing and unbounded continuous function. Let us denote by $(\epsilon_{l.n})_n$ the solutions to the expression: $l(1/\epsilon)= \frac{n \epsilon^2}{8}$ for all $n$. 
	If there is a sequence $(\epsilon_n)_n$ such that:  
	\begin{enumerate}
	\item  $(\epsilon_n)_n \leq (\epsilon_{l.n})_n$ holds eventually with $n$, and 
	\item $ \liminf_{n \rightarrow \infty} \frac{\mathcal{H}_{\epsilon_{n}}( \Lambda_{\tilde{f}_{k_n}})}{l(1/\epsilon_n)} \geq 1,$
	\end{enumerate} 
	then\footnote{In particular, if  $\liminf_{n \rightarrow \infty} \frac{\mathcal{H}_{\epsilon_{l,n}}( \Lambda_{\tilde{f}_{k_n}})}{l(1/\epsilon_{l,n})} \geq 1$ then $\liminf_{n \rightarrow \infty} \frac{R^+( \Lambda^n_{\tilde{f}_{k_n}})}{\log (e) \cdot l(1/\epsilon_{l,n})} \geq 1.$} 
	$$ \liminf_{n \rightarrow \infty} \frac{R^+( \Lambda^n_{\tilde{f}_{k_n}})}{\log (e) \cdot l(1/\epsilon_n)} \geq 1.$$ 
\end{lemma}

The proof is presented in Section~\ref{proof_LB_redun_metric_entropy_dyn}.

The second result characterizes a sufficient condition on $(\epsilon_n)_n$, function of $(k_n)_n$, i.e., the size sequence of tail based partitions, where the metric entropy of our collection of envelope distributions shares the same asymptotic than the unconstrained family determined in (\ref{eq_th_IRG_tail_schemme_11}).

\begin{lemma}\label{proposition_metric_entropy_asymtotics}
	Let us consider a sequence of non-negative integer $(k_n)_n$ and a sequence of non-negative reals $(\epsilon_n)_n$, where $(1/k_n)_n$ is  $o(1)$ and $(\epsilon_n)_n$ is $o(1)$.
	If $f \in \ell_1(\mathcal{X})$ and $(\epsilon_n)_n$ is admissible for $(k_n)_n$ given $\Lambda_f$ (see Def. \ref{admisibility_envelop_given_size})
	then  
	$$\mathcal{H}_{\epsilon_n}(\Lambda_{\tilde{f}_{k_n}})=(1+a_n) \int_1^{1/\epsilon_n^2} \frac{U_f(x)}{2x} \partial x$$ 
	for a 
	sequence $(a_n)_{n}$ being $o(1)$, and consequently,
	\begin{equation} \label{eq_th_IRG_tail_schemme_13}
		 \lim_{n \rightarrow \infty} \frac{\mathcal{H}_{\epsilon_n}(\Lambda_{\tilde{f}_{k_n}})}{l_f(1/\epsilon_n)}= \lim_{n \rightarrow \infty} \frac{\mathcal{H}_{\epsilon_n}(\Lambda_{\tilde{f}_{k_n}})}{\mathcal{H}_{\epsilon_n}(\Lambda_f)}=1 ,
	\end{equation}
	where $l_f(1/\epsilon)=\int_1^{1/\epsilon^2} \frac{U_f(x)}{2x} \partial x$ for $\epsilon>0$ (see Def. \ref{def_function_lf}). 
\end{lemma}

The proof of this result is presented in Section~\ref{proof_proposition_metric_entropy_asymtotics}.

Comments on Lemma~\ref{proposition_metric_entropy_asymtotics}:
\begin{enumerate}
	\item Lemma \ref{proposition_metric_entropy_asymtotics} establishes concrete sufficient conditions where $(\mathcal{H}_{\epsilon_n}(\Lambda_{\tilde{f}_{k_n}}))_n$ has the same asymptotic than the metric entropy of the unconstrained family $(\mathcal{H}_{\epsilon_n}(\Lambda_{f}))_n$, which is $\sim (\int_1^{1/\epsilon^2_n} \frac{U_f(x)}{2x} \partial x)_n$
		from (\ref{eq_th_IRG_tail_schemme_11}). 
	\item The proof of this result follows the {\em volume comparison arguments} proposed by Bontemps in \cite[Lemmas 1 and 2]{bontemps_2011}. 
	\item Note that  if 
	$\lim_{n \rightarrow \infty} \bar{F}_f(k_n)/ \epsilon^2_n=0$ implies that $(\epsilon_n)_n$ is admissible for $(k_n)_n$ given $\Lambda_f$. 
	\item Given $(k_n)_n$ and $f$, $\sqrt{\bar{F}_f(k_n)}$ offers a lower bound on the admissible regime for $(\epsilon_n)_n$ (see Def.\ref{admisibility_envelop_given_size}). 
	\item If $(\tilde{k}_n)_n \gg  ({k_n})_n$, i.e.,  $k_n/\tilde{k}_n \rightarrow 0$ as $n\rightarrow \infty$, then from Definition \ref{admisibility_envelop_given_size} 
	$(\tilde{k}_n)_n$ offers a bigger admissible range for the $(\epsilon_n)_n$  
	than its counterpart for $({k}_n)_n$.  
\end{enumerate} 

Finally, as the asymptotic of the metric entropy in (\ref{eq_th_IRG_tail_schemme_11}) offers a tight lower bound to the information radius of envelope families \cite[Th.2]{bontemps_2014},  Lemma \ref{proposition_metric_entropy_asymtotics} in conjunction with Lemma \ref{lemma_LB_redundancy_metric_entropy_dynamic}  provide the mean to characterize a regime of no gain in minimax redundancy as presented in the proof of Theorem \ref{pro:no_gain_minimax_redundancy} below.  

{\em Proof of Theorem \ref{pro:no_gain_minimax_redundancy}:}
	Using the hypothesis that $(\epsilon_{f,n})_n$ is admissible for $(k_n)_n$  given $\Lambda_f$, we have from Lemma \ref{proposition_metric_entropy_asymtotics} that as $n$ goes to infinity:
	\begin{equation}\label{eq_1_proof_no_gain_minimax_redundancy}
		\mathcal{H}_{\epsilon_{f,n}}(\Lambda_{\tilde{f}_{k_n}}) \geq (1+o(1))  \cdot l_f(1/\epsilon_{f,n}), 
	\end{equation}
	which implies that 
	\begin{IEEEeqnarray}{rCl}
	\liminf_{n \rightarrow \infty} \frac{\mathcal{H}_{\epsilon_{f,n}}(\Lambda_{\tilde{f}_{k_n}})}{l_f(1/\epsilon_{f,n})} &\geq& 1.
		\end{IEEEeqnarray}
	Note that $l_f(x)=\int_1^{x^2} \frac{U_f(\bar{x})}{2\bar{x}} \partial \bar{x}$ on $[1,\infty)$ (Def.\ref{def_function_lf}) is strictly increasing,  continuous and unbounded, consequently applying Lemma \ref{lemma_LB_redundancy_metric_entropy_dynamic} it follows that:
	\begin{align}\label{eq2__proof_no_gain_minimax_redundancy}
		&\liminf_{n \rightarrow \infty} \frac{R^+( \Lambda^n_{\tilde{f}_{k_n}})}{\log (e) \cdot l_f(1/\epsilon_{f,n})} \geq 1\nonumber\\ 
		\Leftrightarrow 
		 &\liminf_{n \rightarrow \infty} \frac{R^+( \Lambda^n_{\tilde{f}_{k_n}})}{\log (e) \cdot n \epsilon^2_{f,n}/8} \geq 1, 
	\end{align}
	where the last identity follows from the definition of $\epsilon_{f,n}$ in Def. \ref{def_epsilon_f}.
	
	At this point we use the result in \cite[Proposition 3]{bontemps_2014} that shows 
	that the function $(l_f(x))$ is {\em very slowly variant} \cite[Def. 4]{bontemps_2014}, in the sense that $\forall \eta\geq 0$ and $\kappa>0$, 
	\begin{equation}\label{eq3__proof_no_gain_minimax_redundancy}
		\lim_{x \rightarrow \infty} \frac{l_f(\kappa  x l_f(x)^\eta)}{l_f(x)}=1.
	\end{equation}
	This slowly variant condition implies that $(\epsilon_{f,n})_n$, as a (point-wise) solution of the condition $l_f(1/\epsilon)=n \epsilon^2/8$,  
	satisfies asymptotically (the argument  is presented in  the proof of \cite[Theorem 5]{haussler_1997}) that: 
	\begin{equation}\label{eq4__proof_no_gain_minimax_redundancy}
		\left( \epsilon_{f,n}^2/8 \right)_n \sim  \left( \frac{l_f(\sqrt{n})}{n} \right)_n = \left( \frac{1}{n} \int_1^{n} \frac{U_f(x)}{2x} \partial \bar{x}) \right)_n.
	\end{equation}
	Consequently  (\ref{eq4__proof_no_gain_minimax_redundancy}) and (\ref{eq2__proof_no_gain_minimax_redundancy}) prove the result. 
\hspace{\fill}~\QED

\subsection{Proof of Lemma \ref{lemma_LB_redundancy_metric_entropy_dynamic}}
\label{proof_LB_redun_metric_entropy_dyn}
\begin{proof}
	From Lemma \ref{lemma_HO_metric_entropy} it follows that $\frac{R^+( \Lambda^n_{\tilde{f}_{k_n}})}{l(1/\epsilon_{l,n})} \geq$ 
	\begin{align}\label{eq_1_proof_LB_redun_metric_entropy_dyn}
		 &\log (e) \min \left\{ \frac{\mathcal{H}_{\epsilon_{l,n}} (\Lambda_{\tilde{f}_{k_n}})}{l(1/\epsilon_{l,n})},\frac{n\epsilon^2}{8 \cdot l(1/\epsilon_{l,n})}  \right\} -\frac{1}{l(1/\epsilon_{l,n})},\nonumber\\
			&= \log (e) \min \left\{ \frac{\mathcal{H}_{\epsilon_{l,n}} (\Lambda_{\tilde{f}_{k_n}})}{l(1/\epsilon_{l,n})}, 1 \right\} -\frac{1}{l(1/\epsilon_{l,n})}
	\end{align}
	for all $n$. 
	As $(\epsilon_{l,n}) \geq (\epsilon_n)$, without loss of generality we assume that there is a mapping 
	$\tau:\mathbb{N} \longrightarrow \mathbb{N}$ such that $\epsilon_{l,n}=\epsilon_{\tau(n)}$ for every  $n$, 
	where $\tau(n) \leq n$ eventually with $n$. From construction $(\epsilon_{l,n})$ is $o(1)$ and thus, 
	$l(1/\epsilon_{l,n}) \longrightarrow \infty$. Then,
	\begin{IEEEeqnarray}{rCl}\label{eq_2_proof_LB_redun_metric_entropy_dyn}
		&\lim& \inf_{n \longrightarrow \infty} \frac{R^+( \Lambda^n_{\tilde{f}_{k_n}})}{l(1/\epsilon_{l,n})} \geq \log (e) \min \left\{ \liminf_{n \longrightarrow \infty} \frac{\mathcal{H}_{\epsilon_{\tau(n)}} (\Lambda_{\tilde{f}_{k_n}})}{l(1/\epsilon_{\tau(n)})}, 1 \right\}\nonumber\\
		&=& \log (e) \cdot \min \left\{ \liminf_{n \longrightarrow \infty} \frac{\mathcal{H}_{\epsilon_{n}} \left( \Lambda_{\tilde{f}_{k_{\tau^{-1}(n)}}} \right)}{l(1/\epsilon_n)}, 1 \right\} \nonumber\\
		\label{eq_2_proof_LB_redun_metric_entropy_dynb}
		&\geq& \log (e) \cdot \left\{ \liminf_{n \longrightarrow \infty} \frac{\mathcal{H}_{\epsilon_{n}} \left( \Lambda_{\tilde{f}_{k_{n}}} \right)}{l(1/\epsilon_n)}, 1 \right\} \\
		\label{eq_2_proof_LB_redun_metric_entropy_dynb2}
		&\geq & \log(e),
	\end{IEEEeqnarray}
	where the  inequality in (\ref{eq_2_proof_LB_redun_metric_entropy_dynb}) follows from the fact that $\mathcal{H}_{\epsilon} \left( \Lambda_{\tilde{f}_k} \right) \leq \mathcal{H}_{\epsilon} \left( \Lambda_{\tilde{f}_{\bar{k}}} \right)$ if $\bar{k}\geq k$ and that  $\tau^{-1}(n) \geq n$, and (\ref{eq_2_proof_LB_redun_metric_entropy_dynb2}) from the main hypothesis of Lemma \ref{lemma_LB_redundancy_metric_entropy_dynamic}.
\end{proof}

\subsection{Proof of Lemma \ref{proposition_metric_entropy_asymtotics}}
\label{proof_proposition_metric_entropy_asymtotics}
\begin{proof}
Following the lower bound for $\mathcal{D}_\epsilon(\Lambda_f)$ elaborated in \cite[Lemma 3]{bontemps_2011}, 
which is based on a  covering argument and volume based inequality, 
we have that  for all $m\geq 1$  and any arbitrary $k$ such that $m+l_f \leq k$ (see Def. \ref{def_tail_function_critical_dimension}), 
\begin{IEEEeqnarray}{rCl}\label{eq_1_proof_prop_metric_entropy_asymt}
	\mathcal{D}_\epsilon (\Lambda_{\tilde{f}_k})  &\geq &\frac{\textrm{Vol} \left( \prod\limits_{i=l_f+1}^{l_f+m} [0,\sqrt{\tilde{f}_k(i)}] \right)}{\textrm{Vol} \left( \mathcal{B}_m(\epsilon) \right)}\nonumber\\
	& = &\frac{\prod\limits_{i=l_f+1}^{l_f+m} \sqrt{\tilde{f}_k(i) } }{\epsilon^m \textrm{Vol}(\mathcal{B}_m)}, 
\end{IEEEeqnarray}
where $\mathcal{B}_m(\epsilon)$ denotes the ball in $\mathbb{R}^m$ of radius $\epsilon$, and $\mathcal{B}_m \coloneqq    \mathcal{B}_m(1)$.
If we consider 
$$N_\epsilon \coloneqq \inf \left\{m\geq 1: \bar{F}_{f}(m) <  \epsilon^2/16\right\}$$
 and let $\tilde{N}^k_\epsilon \coloneqq \min \left\{ N_\epsilon , k \right\}$, and we evaluate (\ref{eq_1_proof_prop_metric_entropy_asymt}) with $m=\tilde{N}^k_\epsilon - l_f$  it follows that  
\begin{align}\label{eq_2_proof_prop_metric_entropy_asymt}
	&\mathcal{H}_\epsilon(\Lambda_{\tilde{f}_k}) = \ln \mathcal{D}_\epsilon (\Lambda_{\tilde{f}_k})  \geq \nonumber\\
	 &\sum_{i=l_f+1}^{\tilde{N}^k_\epsilon} \ln  \sqrt{\tilde{f}_k(i)} - \ln \textrm{Vol} (\mathcal{B}_{\tilde{N}^k_\epsilon-l_f}) -   (\tilde{N}^k_\epsilon-l_f) \ln 1/\epsilon. 
\end{align}
On the other hand,  we can adopt the upper bound in \cite[Lemma 2]{bontemps_2011} that is based on another 
volume comparison argument, leading to 
\begin{IEEEeqnarray}{rCl}\label{eq_3_proof_prop_metric_entropy_asymt}
	\mathcal{D}_\epsilon (\Lambda_{\tilde{f}_k}) &\leq& \frac{\textrm{Vol} \left(\prod\limits_{i=1}^{\tilde{N}^k_\epsilon} \left[-\epsilon/8, \sqrt{\tilde{f}_k(i)} + \epsilon/8\right] \right)}{\textrm{Vol}\left(\mathcal{B}_{\tilde{N}^k_\epsilon}(\epsilon/8) \right)}  \nonumber\\
	&=&\frac{\prod\limits_{i=1}^{\tilde{N}^k_\epsilon} \left(\sqrt{\tilde{f}_k(i)} + \epsilon/4\right)}{ (\epsilon/8)^{\tilde{N}^k_\epsilon} \cdot \textrm{Vol}(\mathcal{B}_{\tilde{N}^k_\epsilon})}.
\end{IEEEeqnarray}
This inequality reduces to \cite[Eq. (6)]{bontemps_2014}
\begin{IEEEeqnarray}{rCl}
	\mathcal{H}_\epsilon(\Lambda_{\tilde{f}_k})& \leq &\sum_{i=1}^{l_f} \ln \left(\sqrt{\tilde{f}_k(i)} + \epsilon/4\right) + \sum_{i=l_f+1}^{\tilde{N}^k_\epsilon} \ln \left(\sqrt{\tilde{f}_k(i)}\right) \nonumber\\
	&-& \ln \left( \textrm{Vol} \left(\mathcal{B}_{\tilde{N}^k_\epsilon}\right) \right) + \frac{\tilde{N}^k_\epsilon-l_f}{\sqrt{1-e^{-b}}} + \tilde{N}^k_\epsilon \ln 8/\epsilon,\label{eq_4_proof_prop_metric_entropy_asymt}
\end{IEEEeqnarray}
for $b=-\ln \bar{F}(l_f)>0$.

	If we consider the regime where $N_\epsilon< k$, then it follows that  $N_\epsilon=\tilde{N}^k_\epsilon$ and  $ \left\{ l_f+1,\dots,N_\epsilon \right\} \subset   \left\{1,\dots,k-1\right\}$. Therefore $\tilde{f}_k(i)=f(i)$ for all $i\in \left\{ l_f+1,..,N_\epsilon \right\}$. In this scenario, we have that:
	\begin{align}\label{eq_5_proof_prop_metric_entropy_asymt}
		\mathcal{H}_\epsilon(\Lambda_{\tilde{f}_k}) &\geq \sum_{i=l_f+1}^{{N}_\epsilon} \ln  \sqrt{{f}(i)} - \ln \textrm{Vol} (\mathcal{B}_{{N}_\epsilon-l_f}) \nonumber\\
		&-   ({N}_\epsilon-l_f) \ln 1/ \epsilon 
	\end{align}
	and
	\begin{align}\label{eq_5b_proof_prop_metric_entropy_asymt}
		\mathcal{H}_\epsilon(\Lambda_{\tilde{f}_k}) &\leq \sum_{i=1}^{l_f} \ln (\sqrt{{f}(i)} + \epsilon/4) + \sum_{i=l_f+1}^{{N}_\epsilon} \ln (\sqrt{{f}(i)}) \nonumber\\ 
		&- \ln \left( \textrm{Vol} (\mathcal{B}_{{N}_\epsilon}) \right) + \frac{{N}_\epsilon-l_f}{\sqrt{1-e^{-b}}} + {N}_\epsilon \ln 8/\epsilon. 
	\end{align}
	We point out that the RHS expression of (\ref{eq_5_proof_prop_metric_entropy_asymt}) and  (\ref{eq_5b_proof_prop_metric_entropy_asymt}) are the very same lower and upper bounds derived in \cite[Eq.(7) and Eq.(6)]{bontemps_2014} for $\mathcal{H}_\epsilon(\Lambda_{f})$, respectively. 
	Consequently in this regime, we obtain the lower and upper bound expressions of the unconstrained (i.e., lossless) problem.
	
	By definition $N_\epsilon< k$ is equivalent to the condition that  $\bar{F}_f(k-1) \leq \frac{\epsilon^2}{16}$.  Using the hypothesis that $(k_n)$ and $(\epsilon_n)$ are such that  the condition in Eq.(\ref{eq_admisibility_condition_given_the_size}) (in Def.\ref{admisibility_envelop_given_size}) 
	is satisfied eventually with $n$, it follows that
	\begin{align}\label{eq_7_proof_prop_metric_entropy_asymt}
		&\liminf_{n \rightarrow \infty} \mathcal{H}_{\epsilon_n}(\Lambda_{\tilde{f}_{k_n}}) \geq \nonumber\\  
		&\liminf_{n \rightarrow \infty}    \left\{ \sum_{i=l_f+1}^{{N}_{\epsilon_n}} \ln  \sqrt{{f}(i)} - \ln \textrm{Vol }(\mathcal{B}_{{N}_{\epsilon_n}-l_f})\right. \nonumber\\
	&\left.-   ({N}_{\epsilon_n}-l_f) \ln \frac{1}{{\epsilon_n}} \right\},
	\end{align}
	and
	\begin{align}\label{eq_7b_proof_prop_metric_entropy_asymt}
		&\limsup_{n \rightarrow \infty} \mathcal{H}_{\epsilon_n}(\Lambda_{\tilde{f}_{k_n}}) \leq \nonumber\\ 
		&\limsup_{n \rightarrow \infty}  \left\{  \sum_{i=1}^{l_f} \ln \left(\sqrt{{f}(i)} + \epsilon/4\right) + \sum_{i=l_f+1}^{{N}_\epsilon} \ln \left(\sqrt{{f}(i)}\right)\right. \nonumber\\
		&\left.- \ln \left( \textrm{Vol} (\mathcal{B}_{{N}_\epsilon}) \right) + \frac{{N}_\epsilon-l_f}{\sqrt{1-e^{-b}}} + {N}_\epsilon \ln \frac{8}{\epsilon}\right\}. 
	\end{align}

To conclude, it has been shown in \cite[Prop. 4]{bontemps_2014} that  the RHS of both (\ref{eq_5_proof_prop_metric_entropy_asymt}) and (\ref{eq_5b_proof_prop_metric_entropy_asymt}) behaves asymptotically  as $(1+o(1))\int_{1}^{1/\epsilon^2} \frac{U_f(x)}{2x} \partial x$ when $\epsilon$ goes to zero. Consequently given that by hypothesis $\epsilon_n \longrightarrow 0$,  this fact implies that  $\mathcal{H}_{\epsilon_n}(\Lambda_{\tilde{f}_{k_n}})=(1+o(1)) \int_1^{1/\epsilon_n^2} \frac{U_f(x)}{2x} \partial x$  as $n$ tends to infinity. Finally  (\ref{eq_th_IRG_tail_schemme_13}) follows from (\ref{eq_7_proof_prop_metric_entropy_asymt}), (\ref{eq_7b_proof_prop_metric_entropy_asymt}) and \cite[Prop. 4]{bontemps_2014}.
\end{proof}

\appendices
\section{Supporting Results}
\subsection{Proof of Lemma \ref{lemma_zero_distortion}}
\label{proof_lemma_zero_distortion}
\begin{proof}
	Let first prove the sufficient condition. Let us assume that $ \left\{\pi_n: n\geq 1 \right\}$
	is asymptotically sufficient for $\mu$. The induced distortion is given by 
	\begin{equation}\label{eq_proof_lemma_zero_distortion_1}
		d(\phi_n, \psi_n,\mu^n) = \sum_{x\in \mathcal{X}}f_\mu(x) \cdot \underbrace{\rho (x,\psi_n(\phi_n(x)))}_{ g_n(x) \coloneqq }.
	\end{equation}
Considering that $\mu(\lim\sup_{n}\pi_n(x))= f_\mu(x)$ for all $x\in \textrm{supp}(\mu)$, then it follows 
that $\lim_{n \rightarrow \infty} \psi_n(\phi_n(x)))=x$, $\mu$-almost everywhere and $\lim_{n \rightarrow \infty }g_n(x)=0$, $\mu$-almost surely. Furthermore,  $g_n(x)$ is a bounded function by definition, then the {\em bounded convergence theorem} \cite{varadhan_2001} implies that $\lim_{n \rightarrow \infty} \int_{\mathcal{X}} g_n(x)d\mu(x)=0 \Leftrightarrow  \lim_{n \rightarrow \infty} d(\phi_n, \psi_n,\mu^n)=0$.

For the converse, let us assume that $\cap_{n\geq 1}\cup_{m\geq n} \pi_m(x) \neq  \left\{x \right\}$, $\mu$-almost surely. In other words, $\exists x,x_o\in \textrm{supp}(\mu)$ with $x\neq x_o$ such that $ \left\{x,x_o\right\} \subset  \lim_{n \rightarrow \infty} \cup_{m\geq n} \pi_m(x)$. Consequently, there exists $N$ such that for all $n\geq N$, $d(\phi_n, \psi_n,\mu^n) \geq \min\left\{f_\mu(x), f_\mu(x_o) \right\} \cdot \min\left\{ \rho(x,x_o), \rho(x_o,x) \right\} >0$. 
\end{proof}

\subsection{Proof of Lemma \ref{lemma_asym_suff}}
\label{proof_lemma_asym_suff}
For the proof we need the following definitions: 
\begin{definition}\label{def_mu_summability}
	Let us consider $\mu\in \mathcal{P}(\mathcal{X})$ and a function $g:\mathcal{X} \rightarrow \mathbb{R}$. 
	$g$ is said to be integrable with respect to $\mu$ if $\sum_{x\in \mathcal{X} }  \left|  g(x) \right|   f_{\mu}(x) < \infty$. 
	Finally, $\ell_1(\mu)$ denotes the collection of all integrable functions with respect to $\mu$. 
\end{definition}

\begin{proof}
	Let us consider: 
	\begin{equation}
		H(\mu) - H_{\sigma(\pi_n)} (\mu) = \sum_{x\in \mathcal{X}} f_\mu(x) \cdot \underbrace{  \log \frac{\mu(\pi_n(x))}{f_u(x)}}_{\tilde{g}_n(x) \coloneqq  }. 
 	\end{equation}
From the assumption that $H(\mu)<\infty$, then $\tilde{g}_n(x) \in \ell_1(\mu)$. Furthermore, under
the assumption that $\left\{\pi_n:n \geq 1 \right\}$ is asymptotically sufficient, we have that $\lim_{n \rightarrow \infty} \mu(\pi_n(x))=f_\mu(x)$ for all $x\in \textrm{supp}(\mu)$ and thus, $\lim_{n \rightarrow \infty} \tilde{g}_n(x) = 0$ $\mu$-almost everywhere. 
Finally considering that $\tilde{g}_n(x) \leq \log  1/f_\mu(x) \in \ell_1(\mu)$, the dominated convergence theorem 
implies that $\lim_{n \rightarrow \infty} \sum_{x\in \mathcal{X}} f_\mu(x) \cdot \tilde{g}_n(x)=0$.
\end{proof}

\subsection{Proof of Proposition \ref{proposition_entropy_two_stage}:} 
\label{proof_proposition_entropy_two_stage}
\begin{proof}
The argument reduces to verify the achievability of the entropy using a two-stage lossy construction.   
For that we consider the tail partition 
\begin{align}\label{eq_proof_lemma_entropy_5}
	\pi_n=\left\{ \left\{1 \right\},\dots,\left\{n \right\}, \left\{n+1,\dots,\right\} \right\}, 
\end{align}
associated to $\phi_n(x)=x$ if $x\in \left\{1,..,n\right\}$ and otherwise $\phi_n(x)=0$. 
It is simple to verify that this scheme satisfies the zero distortion condition. 
For the lossless coding of $Y^n=\Phi_n(X^n)\in \left\{0,1,\dots,n\right\}^n$, we can 
consider the prefix-free {\em Shannon code} \cite{cover_2006}, whose rate is at most two bits away from the 
entropy of  $Y^n$.  Hence,  there is $\mathcal{C}_n:\left\{0,1,..,n\right\}^n \longrightarrow \left\{ 0,1\right\}^*$
such that: 
\begin{equation}\label{eq_proof_lemma_entropy_5}
	r(\phi_n, \mathcal{C}_n, \mu^n) \leq \frac{H(Y^n) +1}{n} = \frac{H_{\sigma(\pi_n)}(\mu) +1}{n},
 \end{equation}
which suffices to show that 
\begin{equation}
\limsup_{n \longrightarrow \infty} r(\phi_n, \mathcal{C}_n, \mu^n) \leq H(\mu),
 \end{equation}
and therefore $R_{al}(\mu) \leq \bar{R}_{al}(\mu)\leq H(\mu)$.
\end{proof}

\subsection{Proof of Lemma \ref{pro_invariance_envelop}}
\label{proof_pro_invariance_envelop}
\begin{proof}
	First, it is direct to show that $\left\{v_\mu: \mu\in \Lambda_f \right\} \subset \tilde{\Lambda}_{\tilde{f}}$. 
	 Then, it remains to prove that for any $v\in \tilde{\Lambda}_{\tilde{f}}$  there is $\mu \in \Lambda_f$
	such that $v_\mu=v$, in total variations. Let us fix an arbitrary $i\in \mathcal{I}$.  
	If we first assume that  $\left| \mathcal{A}_{n,i} \right|< \infty$,  we propose the following  approach: 
	\begin{align}\label{eq_proof_pro_invariance_envelop_1}
		&\hat{x}_1=\arg \min_{x\in \mathcal{A}_{n,i}} f(x),   &\hat{w}_{x_1} =\frac{f(\hat{x}_1)}{\sum_{x\in {\mathcal{A}_{n,i}}}{f}(x)},\nonumber\\
		&\hat{x}_2=\arg \min_{x\in \mathcal{A}_{n,i}\setminus \left\{ \hat{x}_1 \right\}} f(x),  &\hat{w}_{x_2} =\frac{f(\hat{x}_2)}{\sum_{x\in {\mathcal{A}_{n,i}}}{f}(x)},\nonumber\\
		&,\dots ,& 
	\end{align}
	where finally  $\hat{x}_{\left| \mathcal{A}_{n,i} \right|} \in \mathcal{A}_{n,i} \setminus \left\{\hat{x}_i:i=1,\dots, \left| \mathcal{A}_{n,i} \right|-1 \right\}$      and  $\hat{w}_{x_{\left| \mathcal{A}_{n,i} \right|}} =\frac{f(\hat{x}_{\left| \mathcal{A}_{n,i} \right|})}{\sum_{x\in {\mathcal{A}_{n,i}}}{f}(x)}$. With this we define $\mu( \left\{ \hat{x}_i \right\})=\hat{w}_{x_1}\cdot v(\left\{ i \right\} )$ for each $i\in \left\{1,\dots, \left| \mathcal{A}_{n,i} \right| \right\}$. Note that $\mu( \left\{ \hat{x}_i \right\}) \leq f( \hat{x}_i )$  and $\sum_{x\in \mathcal{A}_{n,i}} f_\mu(x)=v(\left\{ i \right\} )$ by construction. 
	If $\left| \mathcal{A}_{n,i} \right|= \infty$ and $\sum_{x\in \mathcal{A}_{n,i}}f(x) < \infty$, we can follow the same inductive approach than in (\ref{eq_proof_pro_invariance_envelop_1}) to construct $\mu(\left\{ x\right\})$ for all  $x\in \mathcal{A}_{n,i}$.  
	On the other hand, if $\left| \mathcal{A}_{n,i} \right|= \infty$ and $\sum_{x\in \mathcal{A}_{n,i}}f(x) = \infty$, then $\tilde{f}(i)=1$ by definition, and we can always  find $\mu_i \in \Lambda_{f}$ such that $\textrm{supp}(\mu_i) \in \mathcal{A}_{n,i}$.  Then, we construct $\mu(\left\{x \right\}) = \mu_i (x) \cdot v(\left\{ i \right\})$, where it is clear that $\mu(\mathcal{A}_{n,i})=v(\left\{ i \right\})$ and  $f_\mu(x) \leq f(x)$ for all $x\in \mathcal{A}_{n,i}$ provided by  $\mu_i \in \Lambda_f$. 
\end{proof}

\subsection{Proposition \ref{pro_redundancy_bound}}
\label{proof_pro_redundancy_bound}
\begin{proposition}\label{pro_redundancy_bound}
	Let us consider a lossy code $(f_n,g_n)$ and a family of distributions $\Lambda\subset \mathcal{P}(\mathcal{X})$. 
	If we denote by $v_{\mu^n}$ the probability in $\mathcal{I}_n$ induced by $\mu^n\in \Lambda^n$ (the $n$-fold distributions with marginal in $\Lambda$) and $\phi_n$, by $v_{\mu^n}(\left\{ i \right\})=\mu^n(\phi_n^{-1}(\left\{ i \right\}))$ $\forall i \in \mathcal{I}_n$,  then
	\begin{IEEEeqnarray}{rCl}
		\bar{R}(\Lambda ,f_n) &=   & \sup_{\mu \in \Lambda} \Big(r(f_n,\mu^n) -\frac{1}{n}H_{\sigma(\pi_n)}(\mu^n)\Big) \nonumber\\
		&\geq &\frac{1}{n} R^{+}(\left\{v_{\mu^n}: \mu\in \Lambda \right\}) \nonumber\\
		& = &   \frac{1}{n} \inf_{v\in \mathcal{P}(\mathcal{I}_n)} \sup_{\mu \in \Lambda}  \mathcal{D}(v_{\mu^n}|v) \nonumber\\
		&=& \frac{1}{n} \inf_{v\in \mathcal{P}(\mathcal{X}^n)} \sup_{\mu \in \Lambda} \mathcal{D}_{\sigma(\pi_n)}(\mu^n | v).
	\end{IEEEeqnarray}
\end{proposition}

\begin{proof}
	By definition $r(f_n,\mu^n)=\frac{1}{n} \mathbb{E}_{X^n_1\sim \mu^n} \left\{ \mathcal{L}(\mathcal{C}_n(\phi_n(X^n_1)))\right\}$.  Consequently,  if we let $Y_n=\phi_n(X^n)$ in $\mathcal{I}_n$ we have that $Y_n\sim v_{\mu^n}$, where 
	$v_{\mu^n}$ denote the probability induced by $\mu^n$ and $\phi_n$ in $\mathcal{P}(\mathcal{I}_n)$. We will consider $r(\mathcal{C}_n, v_{\mu^n})= \mathbb{E}_{Y_n \sim v_{\mu^n}} \left\{ \mathcal{L}(\mathcal{C}_n(Y_n))\right\} =n \cdot r(f_n,\mu^n)$, and as $r(\mathcal{C}_n, v_{\mu^n})\geq H(Y_n)$ \cite{cover_2006}, 
	for the rest we  focus on a refined worst-case redundancy, attributed to the second stage of $f_n$, given by
	\begin{equation}\label{eq_proof_pro_redundancy_bound_1}
		\bar{R}(\Lambda, \mathcal{C}_n) \coloneqq    \sup_{\mu \in \Lambda} \Big(r(\mathcal{C}_n, v_{\mu^n}) - H(v_{\mu^n})\Big).
	\end{equation}
	We note that $H(v_{\mu^n}) = H_{\sigma(\pi_n)}(\mu^n)$, therefore $\bar{R}(\Lambda,f_n)= \frac{1}{n}\bar{R}(\Lambda, \mathcal{C}_n)$. Considering (\ref{eq_proof_pro_redundancy_bound_1}) we have that
	\begin{IEEEeqnarray}{rCl}\label{eq_proof_pro_redundancy_bound_2}
		\bar{R}(\Lambda, \mathcal{C}_n) &\geq &\min_{\tilde{\mathcal{C}}_n:\mathcal{I}_n  \rightarrow \left\{ 0,1\right\}^*} \sup_{\mu \in \Lambda} \big(r(\mathcal{\tilde{C}}_n, v_{\mu^n}) - H(v_{\mu^n})\big)\nonumber\\
									&\geq& \inf_{v \in \mathcal{P}(\mathcal{I}_n)} \sup_{\mu\in \Lambda} \mathcal{D}(v_{\mu^n} | v) \nonumber\\
									&=& R^{+}(\left\{v_{\mu^n}: \mu\in \Lambda \right\})\nonumber\\
									&=& \inf_{v \in \mathcal{P}(\mathcal{X}^n)} \sup_{\mu\in \Lambda} \mathcal{D}_{\sigma(\pi_n)}(\mu^n | v).
	\end{IEEEeqnarray}
	The first inequality in (\ref{eq_proof_pro_redundancy_bound_2}) is because we are solving the least worst-case redundancy (fixing the first stage of $f_n$), the  second is from the tight connection between prefix-free mappings and probabilities in $\mathcal{P}(\mathcal{I}_n)$ and the role of the information divergence in lossless variable length prefix-free coding \cite{csiszar_2004}, and the last equalities are from the definition of the induced probabilities in $\mathcal{P}(\mathcal{I}_n)$ and the identity in (\ref{eq_th_DSU_3}).  We note that the expression in (\ref{eq_proof_pro_redundancy_bound_2}) is the information radius of our $n$-fold family $\Lambda^n$ projected into the sub-sigma field induced by $\pi_n$, i.e., first stage of $f_n$.
\end{proof}
	
\bibliographystyle{IEEEtran}				
\bibliography{main_jorge_silva}			

\end{document}